\documentclass[a4paper,UKenglish,cleveref, autoref, thm-restate]{lipics-v2021}

\usepackage[ruled,vlined]{algorithm2e}
\RestyleAlgo{ruled}
\usepackage{tikz}
\usepackage{mathtools}
\usetikzlibrary{decorations.pathreplacing, decorations.markings}
\usepackage{longtable}
\usepackage{numprint}

\hideLIPIcs

\title{New Approximations for Temporal Vertex Cover on Always Star Temporal Graphs}

\author{Sophia Heck}{Faculty of Computer Science, University of Vienna, Austria}{sophia.heck@univie.ac.at}{https://orcid.org/0009-0001-5971-3051}{}

\author{Eleni Akrida}{Department of Computer Science, Durham University, UK}{eleni.akrida@durham.ac.uk}{https://orcid.org/0000-0002-1126-1623}{}

\authorrunning{ }
\authorrunning{S. Heck and E. Akrida} 

\Copyright{Sophia Heck and Eleni Akrida} 



\acknowledgements{We thank Christian Schulz for helpful inputs and feedback.}

\nolinenumbers 

\ArticleNo{ }

\begin{document}

\newcommand{\tvcsolverrepo}{\url{https://github.com/so-9604/tvc_solver}}
\newcommand{\graphgenrepo}{\url{https://github.com/so-9604/graph_generator}}

\maketitle

\begin{abstract}
Modern real life networks are often highly dynamic. Temporal graphs represent these changes in the network configuration through discrete edge appearances on a fixed set of vertices. All these changes are known in advance until a maximal time step, the lifetime of the graph.
The classical \textit{Vertex Cover} problem  can be naturally extended in the temporal setting to the \textit{Temporal Vertex Cover (TVC)} and \textit{Sliding Window Temporal Vertex Cover (SW-TVC)}. In TVC every edge is covered  by one of their endpoints once over the whole lifetime, while in the~SW-TVC every appearing edge is covered once in every window of~$\Delta$ consecutive time steps.
In always star temporal graphs every snapshot is a star, where the center can switch in every time step.
TVC is proven to be NP-complete on always star temporal graphs, but there is an FPT algorithm parameterized by the sliding window size~$\Delta$, solving it optimally in~$\mathcal{O} (T \Delta(n + m) \cdot 2^{\Delta})$.

In this paper, we provide two polynomial-time approximation algorithms for SW-TVC on always star temporal graph, yielding ~$2\Delta-1$ and $\Delta-1$ approximation ratios in $\mathcal{O} (T)$ and $\mathcal{O} (Tm\Delta^2)$ running time, respectively; these algorithms give exact solutions for $\Delta =1$ and $\Delta\le2$, respectively. We also provide the first -- to our knowledge -- implementation and experimental evaluation of the current state-of the-art approximation algorithms, which provide a~$d$ and~$d-1$-approximation ratio, where~$d$ is the maximal degree of any snapshot.
Our experiments on artificially generated always star temporal graphs show that the new approximation algorithms outperform the known~$d-1$-approximation in terms of running time even in some cases where~$\Delta>d$.
We also test the state-of-the-art algorithms on real life data. We surprisingly conclude that the $d-1$-approximation algorithm performs better in terms of running time than the analytically better $d$-approximation algorithm, when implemented strictly as described in the paper. However, we devise a novel implementation of the $d$-approximation algorithm which runs much faster than the $d-1$-approximation in practice. In any case, the $d-1$-approximation computes smaller solutions.

\end{abstract}

\vspace*{-0.1cm}
\section{Introduction}
\label{sec:introduction}
\vspace*{-0.1cm}
Temporal graphs model systems which change over time by considering an underlying static graph whose edges appear and disappear over time. Many real-life systems, such as biological, social or technical networks are highly time-varying in their behavior, e.g.\ in the spread of diseases,  communication between individuals or the transfer of data~\cite{holme_temporal_2018, michail_elements_2018}.
For such systems, temporal graphs can be used to extract information or solve problems e.g.\ monitoring tasks~\cite{hamm_complexity_2022}. The study of temporal graphs has been discussed frequently in recent literature~\cite{CasteigtsKNP19, BeckerCCKRRZ23, MarinoS23}. 
However, classical known graph problems, e.g.,\ graph coloring or vertex cover, need to be adapted first into this time varying setting~\cite{akrida_temporal_2020, mertzios_sliding_2021, mertzios_computing_2023}.

We focus on Temporal Vertex Cover (TVC) and Sliding Window Temporal Vertex Cover (SW-TVC)~\cite{akrida_temporal_2020}, which are temporal adaptations of the classical Vertex Cover (VC). All of these aim to (temporally) cover all edges of the underlying graph by one of their endpoints.
While VC provides a set of vertex as a solution, temporal adaptations provide a set of \emph{vertex appearances}, which are vertices at defined time steps.
In TVC an edge is considered \textit{temporally covered}, if the solution set contains a vertex appearance, where the vertex is an endpoint of the edge, and the edge appears at that time step. Hence, every edge is covered at one appearance during the entire lifetime of the graph. The goal is then to cover all edges with the minimum set size.
This may not be sufficient for applications where, for example, repeated monitoring is required. Therefore, SW-TVC considers a window of fixed size~$\Delta$ and each appearing edge needs to be covered in every~$\Delta$ consecutive time steps.
A practical example of a use case is the monitoring of tasks in sensor networks, where the sensors are vertices, and an edge appears at a time step when the two sensors communicate. If we know when sensors are supposed to send signals, we might want to check repeatedly if the signals work, without checking every single appearance. We can use SW-TVC to receive a set of vertex appearances to control that. By setting a fixed size,~$\Delta$, of the time window to monitor one can define a timespan, in which every appearing communication is tested at least once.
Both adaptations, TVC and SW-TVC, are known to be NP-complete~\cite{akrida_temporal_2020}.
Therefore, approximation and input restriction are commonly used tools to provide solutions in polynomial time~\cite{akrida_temporal_2020, hamm_complexity_2022}.
Algorithms for star temporal graphs have been considered for problems such as temporal exploration \cite{akrida_temporal_2021}, but not yet for (SW-)TVC. So far, only (approximation) algorithms for (SW-)TVC on (always) path/cycle \cite{hamm_complexity_2022} and always degree at most $d$ \cite{akrida_temporal_2020} temporal graphs have been presented. Here we consider approximations for (SW-)TVC on always star temporal graphs for the first time. 
Moreover, by presenting the first implementations, this paper serves as a pioneering effort in bridging the gap between theory and practice, thereby enabling a variety of applications as star-structures emerge in networks such as body sensor networks \cite{Kwak09}, invitations in online social networks \cite{KumarNT06} or network protocols with elected leaders \cite{Nakano02}.

We present the preliminaries in the next section and related work in the field of temporal graphs in Section~\ref{sec:rel_work}. 
Section~\ref{sec:star_approx} provides our main contribution, the presentation of the two new approximations for always star temporal graphs, together with the proofs for running time and approximation ratios. The two new algorithms achieve approximation ratios of ~$2\Delta-1$ and $\Delta-1$ in $\mathcal{O} (T)$ and $\mathcal{O} (Tm\Delta^2)$, respectively.
We also present a framework allowing computation of ~$\Delta-$TVC with the state-of-the-art $d$~\cite{akrida_temporal_2020} and $d-1$~\cite{hamm_complexity_2022} approximation algorithms, and the two proposed ones. We show through experiments in Section~\ref{sec:exp} that our algorithms provide better approximations faster on always star temporal graphs than the known algorithms. Moreover, on general temporal graphs we test the performance of the state-of-the-art algorithms with real-life instances from the SNAP-library~\cite{leskovec_snap_2014} and show that the analytically better $d$-approximation algorithm performs actually worse than the $d-1$-approximation algorithm in solution size and running time. However, we also present a way of implementing the $d$-approximation such that the analytical expectations are met. 

\vspace*{-0.1cm}
\section{Preliminaries}
\label{sec:fund}

\vspace*{-0.1cm}
\subparagraph*{Graphs and Temporal Graphs.} We consider a graph~$G=(V,E)$ to have a finite set of vertices,~$V$, and a finite set of edges, $E$. Let~$n=|V|$ and $m=|E|$.
We use undirected edges~$e=(u,w)$ where $u,w$ are the endpoints of edge $e$. 
The degree~$d$ of a vertex~$v$ is the number of edges connected to~$v$, $d(v)=\mid\{(u,w)\in E | u=v \text{ or } w=v\}\mid$. Two vertices~$u$,~$w$ are said to be adjacent if~$(u,w) \in E$.
Temporal graphs are graphs whose structure changes over time. Specifically, we consider time to be discrete and allow changes in the edges that are active (i.e. available) in each time step, while the set of vertices remains fixed.
\vspace*{-0.05cm}
\begin{definition}[Temporal graphs]
    \label{def:temp_graph}
    A temporal graph is a pair~$(G, \lambda)$, where~$G = (V , E)$ is an underlying (static) graph and~$\lambda : E \rightarrow 2^\mathbb{N}$ is a time-labeling function which assigns to every edge of~$G$ a set of discrete-time labels.
\end{definition}
\vspace*{-0.05cm}
The timespan of a temporal graph is bounded by a maximal time step~$T$, the so-called \emph{lifetime} of the graph. 
An edge~$e \in E$ is called \textit{active} at time step~$t$ if it appears in that time step, i.e.\ ~$t \in \lambda(e)$.
There are two commonly used ways to represent temporal changes visually, either by considering the subgraph of~$G$ that is active at every time step (\emph{snapshot} of $G$ at time $i$, $G_i$), see Figure~\ref{pic:intro_temp_graph_visual:a}, or by visualizing the underlying structure of the graph and labelling each edge with the time steps in which it is active, see Figure~\ref{pic:intro_temp_graph_visual:b}.
\begin{figure}[!bt]
    \centering
    \subfloat[\footnotesize{Snapshot view of the temporal graph}\label{pic:intro_temp_graph_visual:a}]{\resizebox{0.5\textwidth}{!}{
            \begin{tikzpicture}[main/.style = {draw, circle}]
            \node[draw,thick,minimum width=3cm,minimum height=3.5cm] (box) at (0,0)  {};
            \node[draw,thick,minimum width=3cm,minimum height=0.5cm] (box) at (0,-2)  {1};
            \node[main, minimum size=0.5cm] (a) at (-1,1) {\footnotesize$a$};
            \node[main, minimum size=0.5cm] (b) at (0,1)  {\footnotesize$b$};
            \node[main, minimum size=0.5cm] (c) at (1,0)  {\footnotesize$c$};
            \node[main, minimum size=0.5cm] (d) at (0,-1) {\footnotesize$d$};
        
            \path (a) edge[red] (d);
            \path (b) edge[blue] (a);
        \end{tikzpicture}
        \begin{tikzpicture}[main/.style = {draw, circle}]
            \node[draw,thick,minimum width=3cm,minimum height=3.5cm] (box) at (0,0)  {};
            \node[draw,thick,minimum width=3cm,minimum height=0.5cm] (box) at (0,-2)  {2};
            \node[main, minimum size=0.5cm] (a) at (-1,1) {\footnotesize$a$};
            \node[main, minimum size=0.5cm] (b) at (0,1)  {\footnotesize$b$};
            \node[main, minimum size=0.5cm] (c) at (1,0)  {\footnotesize$c$};
            \node[main, minimum size=0.5cm] (d) at (0,-1) {\footnotesize$d$};
        
            \path (a) edge[red] (d);
            \path (b) edge[violet] (d);
            \path (c) edge[orange] (d);
        \end{tikzpicture}
        \begin{tikzpicture}[main/.style = {draw, circle}]
            \node[draw,thick,minimum width=3cm,minimum height=3.5cm] (box) at (0,0)  {};
            \node[draw,thick,minimum width=3cm,minimum height=0.5cm] (box) at (0,-2)  {3};
            \node[main, minimum size=0.5cm] (a) at (-1,1) {\footnotesize$a$};
            \node[main, minimum size=0.5cm] (b) at (0,1)  {\footnotesize$b$};
            \node[main, minimum size=0.5cm] (c) at (1,0)  {\footnotesize$c$};
            \node[main, minimum size=0.5cm] (d) at (0,-1) {\footnotesize$d$};
        
            \path (c) edge[orange] (d);
            \path (b) edge[blue] (a);
        \end{tikzpicture}

            }
        }
    \hfill
    \subfloat[\footnotesize{Label view of the temporal graph}\label{pic:intro_temp_graph_visual:b}]{\resizebox{0.4\textwidth}{!}{
            \begin{tikzpicture}[main/.style = {draw, circle}]
            \node[main, minimum size=0.5cm] (a) at (-2,1) {\footnotesize$a$};
            \node[main, minimum size=0.5cm] (b) at (0,1)  {\footnotesize$b$};
            \node[main, minimum size=0.5cm] (c) at (2,0)  {\footnotesize$c$};
            \node[main, minimum size=0.5cm] (d) at (0,-1) {\footnotesize$d$};
        
            \path (a) edge[red] node[left]{$1,2$} (d);
            \path (b) edge[blue] node[below]{$1,3$} (a);
            \path (b) edge[violet] node[left]{$2$} (d);
            \path (c) edge[orange] node[below]{$2,3$} (d);
        \end{tikzpicture}
        }
        }
    \caption{Visualizations of temporal graphs}
    \label{pic:intro_temp_graph_visual}
    \vspace*{-0.6cm}
\end{figure}
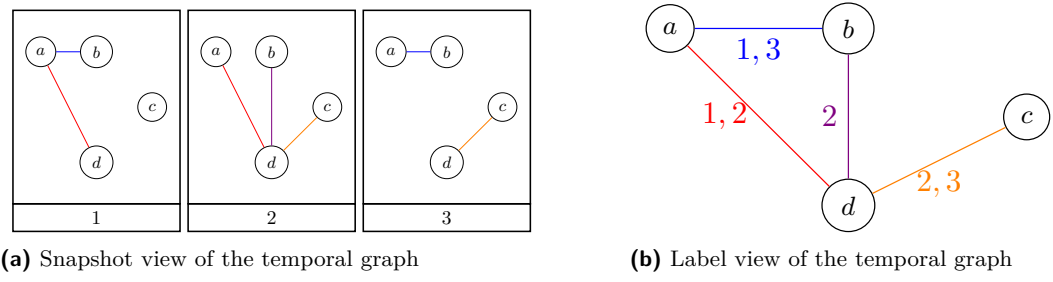

\vspace{-0.3cm}
\subparagraph*{Temporal Vertex Cover.}
In the Temporal Vertex Cover (TVC) problem, we aim to cover each underlying edge by one of its endpoints through a set of vertices at certain time steps. The problem is an adaptation of the classical (static) Vertex Cover (VC) problem, which aims to obtain a minimum-size set of vertices such that at least one endpoint of every edge is included in the set.
The temporal adaptation of VC that aims to cover every edge of the underlying graph (at least) once over the lifetime of the graph. The solution now consists of vertex appearances rather than vertices: a \emph{vertex appearance} or \emph{temporal vertex}~$(u, t)$ describes a vertex~$u \in V$ at a certain time step~$t \in [1, T]$. An edge~$e$ is considered \textit{temporally covered} by a vertex appearance~$(u,t)$, when~$u$ is an endpoint of~$e$ and~$e$ is active at~$t$, i.e.,\ $t \in \lambda(e)$. Hence,~$(u,t)$ covers all edges adjacent to~$u$, which are active at~$t$.
\vspace*{-0.05cm}
\begin{definition}[Temporal Vertex Cover (TVC)~\cite{akrida_temporal_2020}]
    \label{def:tvc}
    Given a temporal graph~$(G, \lambda)$ with lifetime $T$, find a temporal vertex subset~$S\subseteq \{(u,t): u \in V, 1 \leq t \leq T\}$ of~$(G, \lambda)$ of minimum cardinality such that every edge~$e \in E$ is temporally covered by at least one vertex appearance~$(v,t) \in S$.
\end{definition}
\vspace*{-0.05cm}
For most applications, it may not be sufficient to cover each edge only once over the whole lifetime, e.g.,\ if one considers monitoring tasks one wishes a regularly repeated manner of coverage. A commonly used way to face this is to consider a \textit{sliding window}~\cite{akrida_temporal_2020} and look for solutions that satisfy the requirements of the problem within each such window. A window with specific size~$\Delta$ restricts the temporal graph on the specified~$\Delta$ consecutive time steps. It `slides' over the lifetime, such that we have a window~$W_i$ starting in every time step~$i \in [1,T-\Delta]$, see Figure~\ref{pic:intro_valid_swtvc} where the size of the window is two time steps. The time window~$W_t$ can then be defined as the set of time labels starting at~$t$ and covering all time steps until~$t+\Delta-1$.
The TVC requirements can be tightened by specifying that every edge should be temporally covered at least once in every window in which the edge is active.  Let~$E[W_t]$ denote the set of appearing edges in a time window, i.e.,\ $E[W_t] = \{e \in E ~|~ \lambda(e)\cap W_t \neq \emptyset\}$. We say that vertex appearance~$(w,t')$ is in the time window~$W_t$ if~$t'\in W_t$.
\vspace*{-0.05cm}
\begin{definition}[Sliding Window Temporal Vertex Cover (SW-TVC)~\cite{akrida_temporal_2020}]
    \label{def:swtvc}
    Given a temporal graph~$(G, \lambda)$ with lifetime $T$, find a sliding~$\Delta$-window temporal vertex cover of~$(G, \lambda)$ of minimum size.
    That is, find a temporal vertex subset~$S\subseteq \{(u,t): u \in V, 1 \leq t \leq T\}$ of~$(G, \lambda)$ of minimum cardinality such that for every time window~$W_t$ and for every edge~$e \in E[W_t]$,~$e$ is temporally covered by a vertex appearance~$(w,t') \in S$ in~$W_t$.
\end{definition}
\vspace*{-0.05cm}
Through the parameterization with~$\Delta$, the sliding window model is more versatile and even includes the TVC as the special case~$\Delta=T$, as this would be equivalent to considering the problem over the total lifetime~$T$.
When the sliding window size is a fixed constant $\Delta$, we also refer to SW-TVC as $\Delta$-TVC (i.e.,~$\Delta$ is now a part of the problem name). An example of a 2-TVC is shown in Figure~\ref{pic:intro_valid_swtvc}. Since the edge~$(a,b)$ appears non-overlapping in both windows, we need to cover it separately e.g. by including $(a,1)$ and $(a,3)$ in the solution.

\begin{figure}[!bt]
    \center
    \resizebox{0.55\textwidth}{!}{
        \begin{tikzpicture}[main/.style = {draw, circle}]
        \begin{scope}[xshift=-100]
            \node[draw,thick,minimum width=3cm,minimum height=3.5cm] (box) at (0,0)  {};
            \node[draw,thick,minimum width=3cm,minimum height=0.5cm] (box) at (0,-2)  {1};
            \node[main, minimum size=0.5cm, color=white, fill=gray] (a) at (-1,1) {\footnotesize$a$};
            \node[main, minimum size=0.5cm] (b) at (0,1)  {\footnotesize$b$};
            \node[main, minimum size=0.5cm] (c) at (1,0)  {\footnotesize$c$};
            \node[main, minimum size=0.5cm] (d) at (0,-1) {\footnotesize$d$};
    
            \path (a) edge (d);
            \path (b) edge (a);
        \end{scope}
        \begin{scope}[xshift=0]
            \node[draw,thick,minimum width=3cm,minimum height=3.5cm] (box) at (0,0)  {};
            \node[draw,thick,minimum width=3cm,minimum height=0.5cm] (box) at (0,-2)  {2};
            \node[main, minimum size=0.5cm] (a) at (-1,1) {\footnotesize$a$};
            \node[main, minimum size=0.5cm] (b) at (0,1)  {\footnotesize$b$};
            \node[main, minimum size=0.5cm] (c) at (1,0)  {\footnotesize$c$};
            \node[main, minimum size=0.5cm, color=white, fill=gray] (d) at (0,-1) {\footnotesize$d$};
    
            \path (a) edge (d);
            \path (b) edge (d);
            \path (c) edge (d);
        \end{scope}
        \begin{scope}[xshift=100]
            \node[draw,thick,minimum width=3cm,minimum height=3.5cm] (box) at (0,0)  {};
            \node[draw,thick,minimum width=3cm,minimum height=0.5cm] (box) at (0,-2)  {3};
            \node[main, minimum size=0.5cm, color=white, fill=gray] (a) at (-1,1) {\footnotesize$a$};
            \node[main, minimum size=0.5cm] (b) at (0,1)  {\footnotesize$b$};
            \node[main, minimum size=0.5cm] (c) at (1,0)  {\footnotesize$c$};
            \node[main, minimum size=0.5cm] (d) at (0,-1) {\footnotesize$d$};
    
            \path (c) edge (d);
            \path (b) edge (a);
        \end{scope}
    
        \node[draw,thick,red,minimum width=7cm,minimum height=4.5cm] (box) at (-1.75,-0.35)  {};
        \node[red] (label) at (-1.75,-2.85)  {$W_1$};
    
        \node[draw,thick,red,minimum width=7cm,minimum height=4.5cm] (box) at (1.75,-0.15)  {};
        \node[red] (label) at (1.75,2.3)  {$W_2$};
    \end{tikzpicture}
    }

    \caption{Sliding window temporal vertex cover with window size $\Delta=2$ on a graphs with lifetime $T=3$. The set of temporal vertices $\{(a,1), (d,2), (a,3)\}$ provides a valid 2-TVC of minimum size where all edge appearances are temporally covered in every time window.}
    \label{pic:intro_valid_swtvc}
    \vspace*{-0.5cm}
\end{figure}
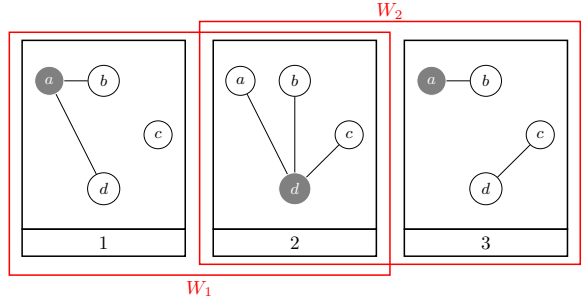

SW-TVC is known to be NP-hard~\cite{akrida_temporal_2020}. Techniques to find solutions for such problems in polynomial time are to approximate the solution or to restrict the inputs to a specific graph class in order to achieve better results. By using the properties of restricted inputs the solution quality and/or running time can be improved.
In terms of restriction there are two general approaches of adapting static graph classes into a temporal setting. For a class~$\mathcal{X}$ of static graphs, a temporal graph $(G,\lambda)$ is called \textit{$\mathcal{X}$ temporal graph} or \textit{underlying~$\mathcal{X}$ temporal graph} if~$G \in \mathcal{X}$, and it is called \textit{always~$\mathcal{X}$ temporal graph}, if~$G_i \in \mathcal{X}$ for every~$i \in [T] = \{1,2,\dots,T\}$~\cite{akrida_temporal_2020}. In this paper we focus on always star temporal graphs, where every snapshot is a star graph, but the center may vary in each time step and always at most degree~$d$ temporal graphs, where every snapshot has degree at most~$d$.
\vspace*{-0.3cm}

\subparagraph*{Approximation Algorithms.} In approximation algorithms, for any input~$I$ the quality of the solution $x$ never exceeds the optimal solution $x^\star$ by the approximation ratio~$\rho=\sup_{I} \frac{f(x( I ))}{f (x^\star ( I ))}$. Quality is determined by the objective function~$f$; for (SW-)TVC $f$ computes the cover size.
\vspace*{-0.2cm}

\section{Related Work}
\label{sec:rel_work}
In the literature, temporal graphs appear under different names, which can also refer to different considered models. We use a popular model where the set of vertices is fixed over time and the edges appear and disappear over time. Besides \textit{temporal}~\cite{akrida_temporal_2020, hamm_complexity_2022} such graphs are also called \textit{evolving}~\cite{bui-xuan_computing_2003, ferreira_building_2004}, \textit{dynamic}~\cite{bhattacharya_deterministic_2018, giakkoupis_randomized_2014,hanauer_recent_2022, MaliMPZ13} or \textit{time-varying}~\cite{casteigts_time-varying_2012, tang_small-world_2010, viard_computing_2016}.
Additionally to path-related temporal problems, which have been studied much in literature ~\cite{akrida_temporal_2021, bui-xuan_computing_2003, fluschnik_temporal_2020, kempe_connectivity_2002, michail_traveling_2016, OettershagenMK22, wu_path_2014}, non-path related temporal problems have recently received more attention; examples include temporal graph coloring~\cite{mertzios_sliding_2021},~$\Delta$-cliques~\cite{viard_computing_2016}, temporal spanners~\cite{casteigts_temporal_2021} and the temporal vertex cover~\cite{akrida_temporal_2020, hamm_complexity_2022}.
The literature considers different techniques of adapting classical problems into a temporal setting.
In particular, there are three commonly used ones: the consideration of the problem over the whole lifetime~\cite{akrida_temporal_2020, mertzios_sliding_2021}; the \emph{sliding window} technique~\cite{akrida_temporal_2020,mertzios_sliding_2021,viard_revealing_2015,viard_computing_2016}; and the maintaining of a correct solution in every time step~\cite{bhattacharya_deterministic_2018, hanauer_recent_2022}. The first two have already been mentioned in Section~\ref{sec:introduction} for VC, yielding TVC and SW-TVC, respectively.
The third technique aims to minimize the update time while guaranteeing an optimal or \hbox{good solution~\cite{bhattacharya_deterministic_2018, hanauer_recent_2022}.}
Bhattacharya et al.~\cite{bhattacharya_deterministic_2018} propose an algorithm for maintaining VC and the maximum matching during every time step.
Such an adaptation provides a correct problem solution for every subgraph.
If all changes over the lifetime are known in advance, solving a problem with this approach is equivalent to SW-TVC, where~$\Delta=1$. Even so in this adaptation, knowing the changes in the network structure is not necessary and one only needs to receive edge updates in order to compute a new solution from the previous one.

The classic VC is known to be NP-hard \cite{GareyJ79}. In general this remains true for the adaptations TVC and SW-TVC~\cite{akrida_temporal_2020}. However, there are better proven results on restricted inputs. While on underlying paths or cycles TVC is solvable in polynomial time~\cite{hamm_complexity_2022}, it remains NP-hard on star temporal graphs~\cite{akrida_temporal_2020}. SW-TVC also remains NP-hard for path or cycle temporal graphs if~$T\ge\Delta\ge2$~\cite{hamm_complexity_2022}.
For arbitrary temporal graphs there are two exact algorithms known for~SW-TVC; both using  dynamic programming. The first one by Akrida et al.~\cite{akrida_temporal_2020} yields a running time in~$\mathcal{O}(T \Delta(n + m) \cdot 2^{n(\Delta+1)})$. For always star temporal graphs this can be bounded to~$\mathcal{O} (T \Delta(n + m) \cdot 2^{\Delta})$.
The second one by Hamm et al.~\cite{hamm_complexity_2022} runs in~$\mathcal{O}(Tc^{\mathcal{O}|E|)})$, where~$c=min\{2^{\mathcal{O}(m)}, \Delta\}$.
This leads to a FPT algorithm for SW-TVC, which parameterized by the size of a solution. This approach also provides the possibility to only solve a partial graph input, which is used by the~$d-1$-approximation~\cite{hamm_complexity_2022} for always degree at most~$d$ temporal graphs, which is implemented in this paper.
Similarly to VC having several approximation algorithms~\cite{abu-khzam_recent_2022, haider_smart_2020, lamm_finding_2017}, literature also provides different ideas of approximating SW-TVC.
Based on the idea that SW-TVC can be reduced to Set Cover Akrida et al.~\cite{akrida_temporal_2020} present two approximation algorithms. By using Linear Programming for Set Cover, proposed by Vazirani~\cite{vazirani_approximation_2001}, this leads to a~$2k$-approximation for SW-TVC~\cite{akrida_temporal_2020}, where~$k$ is the maximum edge frequency defining the maximum appearance of the edge during an arbitrary window. Since~$k$ is maximal $\Delta$, this is also a $2\Delta$-approximation.
Instead of Linear Programming it is also possible to use a greedy approach for Set Cover from Duh and Fürer \cite{duh_approximation_1997} resulting in \hbox{a~$(\ln n + \ln \Delta + \frac{1}{2})$-approximation for  SW-TVC ~\cite{akrida_temporal_2020}.}

For always degree at most~$d$ temporal graphs, i.e. the degree at every snapshot is at most~$d$, two approximation algorithms for SW-TVC are provided.
Akrida et al.~\cite{akrida_temporal_2020} propose a~$d$-approximation with running time in~$\mathcal{O}(mT)$ where~$m$ is the number of edges in the underlying graph~$G$. The algorithm uses the idea to calculate SW-TVC on every possible single-edge temporal subgraph exactly in~$\mathcal{O}(T)$ for every edge and \hbox{take the union of the result.}
Another approximation algorithm by Hamm et al.~\cite{hamm_complexity_2022} is based on the approach to iteratively cover paths with two edges instead of single edges and chooses the middle vertex to be in the cover. With that idea an overall approximation ratio of~$(d-1)$ can be achieved. The running time of this is in~$\mathcal{O}(m^2T^2)$ where~$m$ is the number of edges \hbox{in the underlying graph~$G$.}
As far as the authors know this paper provides a first implementation and experimental evaluation of of both of these algorithms. These algorithms also build the current state of the art for SW-TVC for the restriction to always star temporal graphs.
\vspace*{-0.1cm}

\section{SW-TVC on Always Star Temporal Graphs}
\label{sec:star_approx}

Always star temporal graphs consist of a star graph in every snapshot.
Note that in always star temporal graphs, a solution to SW-TVC includes at most one vertex for each time step, namely the star center at that time step.
The current state of the art to approximate SW-TVC on these graphs, would be to use the~$d$- and the~$d-1$-approximation algorithms that were proposed for always degree at most $d$ temporal graphs \cite{akrida_temporal_2020, hamm_complexity_2022}. We present here a data structure to hold temporal graphs and new approximations using the structural knowledge of always star temporal graphs to achieve better results.
\vspace*{-1em}

\subsection{Data Structure for Temporal Graphs}
To ensure the efficiency of the proposed approximation algorithms below, we shall use a data structure that stores temporal graphs which satisfies the following two requirements:
\begin{itemize}
    \item Provides access to all edges of a particular time step in~$\mathcal{O}(1)$ time;
    \item Provides access to all underlying edges of a vertex in~$\mathcal{O}(1)$ time.
\end{itemize}
To achieve this our data structure holds an array of all the underlying edges of the form $e=(u,v)$ along with their appearances $\lambda(e)$ in the form of an array $\mathcal{E}$ of time labels indicating when $e$ is active.
To access all edge appearances at a certain time step, a vector $\mathcal{T}$ of length~$T$ stores all edge indices active at each time step. This can lead to edge indices being referred to at multiple time steps.
To access all adjacent edges of a vertex, the data structure holds an adjacency list~$\mathcal{A}$, in which the respective edge indices can be retrieved for each vertex.
Every edge index is listed twice in~$\mathcal{A}$, once at each endpoint, as the graph is undirected.
We achieve our requirements and store the temporal graphs efficiently as the space complexity is in~$\mathcal{O}(\mathcal{|E|}) + \mathcal{O}(\mathcal{T} + \mathcal{O}(\mathcal{|A|}) = \mathcal{O}((2+T)\cdot m)  + \mathcal{O}(Tm)= \mathcal{O}(n+Tm) +\mathcal{O}(n+2m)$.
\vspace*{-1em}

\newcommand{\algstartriv}{\texttt{STAR-SC}}
\subsection{Trivial Algorithm \algstartriv{}}

The trivial idea to solve \hbox{(SW-)TVC} for always star temporal graphs is to include the star center in the cover in every time step in which at least one edge is active. To detect the star center for time steps with at least two active edges, we compare arbitrary two edges to identify the common vertex. For time steps where only one edge~$e=(v,w)$  is active, either endpoint of the edge may be considered as the center of the star. In this case we use~$v$ as star center. We call this algorithm \algstartriv{}, where \texttt{SC} stands for Star Center.

\begin{theorem}
    \label{thm:star_trivial}
    The \algstartriv{} algorithm approximates~SW-TVC with window size $\Delta$ on always star temporal graphs where~$T\ge \Delta$ and~$\Delta \ge 2$ with ratio~$2\Delta -1$ in~$\mathcal{O}(T)$ time.
\end{theorem}
\vspace*{-1em}

\begin{proof}
    To prove Theorem~\ref{thm:star_trivial}, we need to prove the running time and approximation ratio.
    The running time of the algorithm is in~$\mathcal{O} (T)$, since we check all time steps and the detection of the star center is in~$\mathcal{O}(1)$, as we compare at most two edges.
    For the approximation ratio, we need to consider the worst case and compare the approximate solution computed by \algstartriv{} to the optimal solution.
    The worst case scenario holds when all edges are active in all time steps, in that case our algorithm includes the (static) star center in every time step, while only one coverage per window is required for an optimal solution.
    Our algorithm cannot perform worse than on a static instance: there is no case in which our algorithm calculates a larger solution or where a smaller optimal solution exists. Since by definition, our algorithm adds at most one node appearance per time step, the most (possibly unnecessary) star centers are added to a solution if at least one edge occurs in each time step. The maximum overlap of edge appearances in windows and thus the smallest optimal solution size is achieved in the static case.    
    Formally, let the size of the exact (i.e. smallest possible)~SW-TVC solution be~$x^\star$ and the size of our approximate solution be~$x$. Then~$x^\star = \lfloor \frac{T}{\Delta} \rfloor$ and~$x = |\mathcal{X}| = T$. We need to show that the approximation ratio~$\frac{x}{x^\star}$ is bounded:$
    \frac{T}{\lfloor \frac{T}{\Delta} \rfloor} \le 2\Delta -1$.
    To break this down, we consider a representation of the lifetime in terms of the window size:~$T = c\cdot\Delta +d$, where~$c,d \in \mathbb{N}$,~$c\geq1$ and~$0\leq d \leq \Delta-1$. Then, we can derive the modulo classes~$R_d \in \{R_0, \dots, R_{\Delta-1}\}$ for the denominator of the equation. Since the most round-off is achieved in the~$R_{\Delta-1}$ class, a value in this class maximizes the ratio~$\frac{T}{\lfloor\frac{T}{\Delta} \rfloor}$. In this class $T$ can be represented as~$T = c\cdot\Delta + (\Delta-1) = a\cdot\Delta -1$ with~$a=c+1, a>1$ and hence~$\lfloor \frac{a\cdot\Delta -1}{\Delta} \rfloor = a-1$.
    To derive the maximum value of $\frac{T}{\lfloor \frac{T}{\Delta} \rfloor}$, we consider any value~$T = a\cdot\Delta -1$ and show that the next larger element of the class~$R_{\Delta-1}$, represented as~$T = (a+1)\cdot\Delta -1$, does not lead \hbox{to a larger approximation ratio.}
    \begin{equation}
        \nonumber
         \begin{aligned}
        \frac{a\cdot\Delta -1}{a-1} \ge \frac{(a+1)\cdot\Delta -1}{a} &\rightleftharpoons
        a\cdot\Delta -1 \ge \frac{(a^2-1)\cdot\Delta -a+1}{a} \rightleftharpoons\\
        a\cdot\Delta -1 \ge  a\cdot\Delta -1 - \frac{\Delta-1}{a}  &\rightleftharpoons
        0 \ge 1-\Delta,
         \end{aligned}
    \end{equation}
    which clearly holds, since~$\Delta\geq2$. Therefore, the~$T$ with the smallest valid value of~$a$, which is~$a=2$, leads to the maximum value of~$\frac{T}{\lfloor \frac{T}{\Delta} \rfloor}$:
    \begin{equation}
        \nonumber
        \frac{2\Delta-1}{\lfloor \frac{2\Delta-1}{\Delta} \rfloor} =\frac{2\Delta-1}{1}=2\Delta-1 \text{ .}
    \end{equation}
    Hence, \algstartriv{} achieves an approximation ratio of~$2\Delta-1$.
\end{proof}

For the case~$\Delta=1$, the algorithm provides the exact solution. This makes sense, since in~$1$-TVC every snapshot is considered separately and the exact solution consists of \emph{every} star center, similar to the solution returned by \algstartriv{}.
\vspace*{-1em}

\newcommand{\algstaradv}{\texttt{STAR-ACOV}}
\subsection{Algorithm \algstaradv{}}
The second algorithm for always star temporal graphs \algstaradv{}, where \texttt{ACOV} stands for All COVered (Algorithm~\ref{algo:star_advanced}), is based on the idea that the star center of a single time step $t$ only needs to be included in the cover if there is any active edge at $t$, which can not be covered by a star center of another time step in all windows containing $t$. This is achieved by checking every sliding window individually. If all edges in a certain time step~$t$ can be covered by other appearances in the window, the star center at time $t$ does need not be added to the cover for that window. However, it still may be needed for a later window.
We use a monitoring table~$\mathcal{C}$ of size $\Delta\times m$ to store one window with all its time steps and for all edges, whether they are active at each of these time steps. To have minimal update costs, we use this table as a ring buffer, in every iteration one time step gets overwritten to switch to the next window and we keep a pointer $first$ to the first time step in current window.
An additional vector~$\mathcal{I}$ stores for each time step in the current window, whether the center vertex is already included (value 2), available (value 1) or excluded (value 0) from the cover.
The goal is to iteratively go through all windows and check if any star center can be excluded from the cover. Therefore, we firstly switch to the current window by updating~$\mathcal{C}$ and~$\mathcal{I}$ based on $first$ (line 4/5) and then iterate over the time steps in the window. If a star center at a time step is not already included in the cover and any of its edges can not be covered by other appearances, we add the star center to our solution and mark it as included. Otherwise, we exclude it from the cover in that window. In that case we need to include the star centers of other available time steps $\mathcal{J}$ which we need for a proper solution. Each  time step~$j\in\mathcal{J}$ is chosen optimal (line 16) by either choosing the latest possible one to cover the respective edge appearance or an earlier one, which is already included in the cover.

\begin{algorithm}[!bt]
    \setcounter{AlgoLine}{0}
    \LinesNumbered
    \KwIn{A temporal graph $(G, \lambda)$ with lifetime $T$, where $G = (V , E)$, and a natural $\Delta \leq T$}
    \KwOut{A temporal vertex cover $\mathcal{X}$ of $(G, \lambda)$ }
    $\mathcal{X} := \emptyset $;
    $\mathcal{C}[\Delta][m] := \langle \langle0,\dots,0\rangle, \dots, \langle0,\dots,0\rangle\rangle $;
    $\mathcal{I}[\Delta] = \langle 1,\dots,1 \rangle$;
    $first=\Delta-1$

    Init $\mathcal{C}$ for timesteps $[1,\Delta-1]$
    
    \ForEach{$t$ in $[1, T-\Delta+1]$}{
        Update $\mathcal{C}$ and $\mathcal{I}$ for timestep $t+\Delta-1$ at $first$

        Update $first$

        \ForEach{$i$ in $[0, \Delta-1]$}{

            $idx_i= (first+i)\%\Delta$

            \If{$\mathcal{I}[idx_i]$ is already included in cover}{
                continue
            }

            \eIf{Any $m \in \mathcal{C}[idx_i]$ is not covered by another (not excluded) star center in $W_t$}{
                $ti = t+i$

                $\mathcal{X} = \mathcal{X} \cap \{(center_{ti},ti)\}$

                $\mathcal{I}[idx_i] = 2$

            }{
                $\mathcal{I}[idx_i] = 0$

                \ForEach{optimal $j$ needed to cover an edge $m_i$ in $i$}{
                    $tj = t+j$

                    $idx_j= (first+j)\%\Delta$

                    $\mathcal{X} = \mathcal{X} \cap \{(center_{tj},tj)\}$

                    $\mathcal{I}[idx_j] = 2$
                }
            }
        }
    }

    \Return{$\mathcal{X}$}
    \caption{\algstaradv{}}
    \label{algo:star_advanced}
\end{algorithm}

\begin{theorem}
    \label{thm:star_advanced}
    The \algstaradv{} algorithm approximates~SW-TVC with window size $\Delta$ on always star temporal graphs where~$T\ge \Delta$ and~$\Delta \ge 2$ with ratio~$\Delta -1$ in~$\mathcal{O}(Tm\Delta^2)$ time.
\end{theorem}
\begin{proof}
    The running time of \algstaradv{} is in~$\mathcal{O}(Tm\Delta^2)$, since the loops in lines 3 and 6 take time~$\mathcal{O}(T\Delta)$. Checking if any edge is not covered by another star center (line 10) takes at most~$\mathcal{O}(m\Delta)$, since we need to look at all edges in every time step in the window. The loop in line 16 takes time at most~$\mathcal{O}(m)$, since we store the cover candidates of every edge separately, to choose the optimal candidate.
    Hence, the overall running time is in~$\mathcal{O}(Tm\Delta^2)$.

    The approximation ratio of the algorithm is a result of the fact that the algorithm excludes the first possible star center appearance that can be covered though others, even if several others could be excluded later were we to keep that star center appearance in the cover.
    
    In order to analyze the approximation ratio, let us first see when the worst case instance occurs for $\Delta>1$. First, suppose that the optimal solution contained less than one star center per window; this would imply that there is no edge appearance in at least one window, and our algorithm would not include these star centers in the solution either. Therefore, we can assume that in the worst case instance, every window has at least one edge appearance and the optimal solution includes at least one star center per window. 
    When the optimal solution includes only one star center per window, that means there exists a star center which covers all appearing edges in that window, hence our algorithm will by definition exclude at least one star center of every window. However, it is possible that all but one star center in every window will be included in our solution (worst case).
    In the case where the optimal solution contains $1\le i < \Delta$ star centers per window, our algorithm will by definition still exclude at least one star center per window. This is not the worst case, since our solution could still include at most all but one star center in every window, but the optimal solution is larger in this case than in the worst case ($i\ge2$ star centers per window).
    If there are $\Delta$ star centers per window in the optimal solution, our algorithm will also include all of them.

    Therefore, the worst case scenario in terms of approximation ratio arises on temporal topologies such as the one shown in Figure~\ref{pic:own_adv_worst}.
    In a general instance of the considered topology the lifetime $T$ is a multiple of $\Delta$ and snapshots repeat every~$\Delta$ time steps, i.e., if~$G_t$ is the snapshot at time step $t$, then $G_t=G_{t+\Delta}$; moreover, the snapshots~$G_2, \dots, G_{\Delta}$ are distinct in their edges and \hbox{$E(G_1) = \bigcup_{i\in[2,\Delta]} E(G_i)$}.
    
    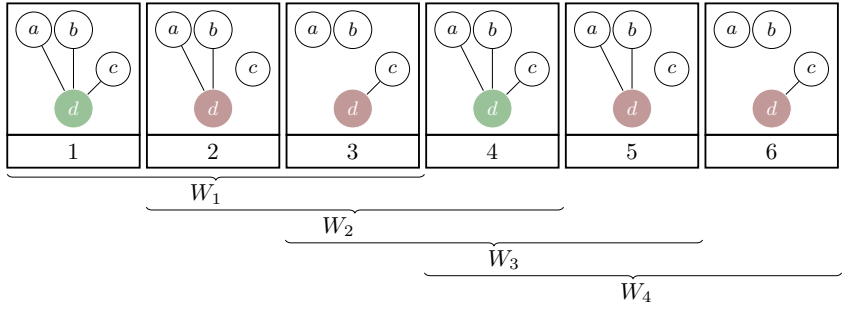
\begin{figure}[!bt]
        \center
        \resizebox{0.8\textwidth}{!}{         
        \begin{tikzpicture}[main/.style = {draw, circle}]
            \begin{scope}[xshift=0]
                \node[draw,thick,minimum width=2cm,minimum height=2cm] (box) at (0,0)  {};
                \node[draw,thick,minimum width=2cm,minimum height=0.5cm] (box) at (0,-1.25)  {1};
                \node[main, minimum size=0.5cm] (a) at (-0.6,0.6){\footnotesize$a$};
                \node[main, minimum size=0.5cm] (b) at (0,0.6)  {\footnotesize$b$};
                \node[main, minimum size=0.5cm] (c) at (0.6,0)  {\footnotesize$c$};
                \node[main, minimum size=0.5cm, color=white, fill=green!40!black!40] (d) at (0,-0.6) {\footnotesize$d$};
        
                \path (a) edge (d);
                \path (b) edge (d);
                \path (c) edge (d);
            \end{scope}
            \begin{scope}[xshift=60]
                \node[draw,thick,minimum width=2cm,minimum height=2cm] (box) at (0,0)  {};
                \node[draw,thick,minimum width=2cm,minimum height=0.5cm] (box) at (0,-1.25)  {2};
                \node[main, minimum size=0.5cm] (a) at (-0.6,0.6){\footnotesize$a$};
                \node[main, minimum size=0.5cm] (b) at (0,0.6)  {\footnotesize$b$};
                \node[main, minimum size=0.5cm] (c) at (0.6,0)  {\footnotesize$c$};
                \node[main, minimum size=0.5cm, color=white, fill=red!40!black!40] (d) at (0,-0.6) {\footnotesize$d$};
        
                \path (a) edge (d);
                \path (b) edge (d);
            \end{scope}
            \begin{scope}[xshift=120]
                \node[draw,thick,minimum width=2cm,minimum height=2cm] (box) at (0,0)  {};
                \node[draw,thick,minimum width=2cm,minimum height=0.5cm] (box) at (0,-1.25)  {3};
                \node[main, minimum size=0.5cm] (a) at (-0.6,0.6){\footnotesize$a$};
                \node[main, minimum size=0.5cm] (b) at (0,0.6)  {\footnotesize$b$};
                \node[main, minimum size=0.5cm] (c) at (0.6,0)  {\footnotesize$c$};
                \node[main, minimum size=0.5cm, color=white, fill=red!40!black!40] (d) at (0,-0.6) {\footnotesize$d$};
        
                \path (c) edge (d);
            \end{scope}
        
            \begin{scope}[xshift=180]
                \node[draw,thick,minimum width=2cm,minimum height=2cm] (box) at (0,0)  {};
                \node[draw,thick,minimum width=2cm,minimum height=0.5cm] (box) at (0,-1.25)  {4};
                \node[main, minimum size=0.5cm] (a) at (-0.6,0.6){\footnotesize$a$};
                \node[main, minimum size=0.5cm] (b) at (0,0.6)  {\footnotesize$b$};
                \node[main, minimum size=0.5cm] (c) at (0.6,0)  {\footnotesize$c$};
                \node[main, minimum size=0.5cm, color=white, fill=green!40!black!40] (d) at (0,-0.6) {\footnotesize$d$};
        
                \path (a) edge (d);
                \path (b) edge (d);
                \path (c) edge (d);
            \end{scope}
            \begin{scope}[xshift=240]
                \node[draw,thick,minimum width=2cm,minimum height=2cm] (box) at (0,0)  {};
                \node[draw,thick,minimum width=2cm,minimum height=0.5cm] (box) at (0,-1.25)  {5};
                \node[main, minimum size=0.5cm] (a) at (-0.6,0.6){\footnotesize$a$};
                \node[main, minimum size=0.5cm] (b) at (0,0.6)  {\footnotesize$b$};
                \node[main, minimum size=0.5cm] (c) at (0.6,0)  {\footnotesize$c$};
                \node[main, minimum size=0.5cm, color=white, fill=red!40!black!40] (d) at (0,-0.6) {\footnotesize$d$};
        
                \path (a) edge (d);
                \path (b) edge (d);
            \end{scope}
            \begin{scope}[xshift=300]
                \node[draw,thick,minimum width=2cm,minimum height=2cm] (box) at (0,0)  {};
                \node[draw,thick,minimum width=2cm,minimum height=0.5cm] (box) at (0,-1.25)  {6};
                \node[main, minimum size=0.5cm] (a) at (-0.6,0.6){\footnotesize$a$};
                \node[main, minimum size=0.5cm] (b) at (0,0.6)  {\footnotesize$b$};
                \node[main, minimum size=0.5cm] (c) at (0.6,0)  {\footnotesize$c$};
                \node[main, minimum size=0.5cm, color=white, fill=red!40!black!40] (d) at (0,-0.6) {\footnotesize$d$};
        
                \path (c) edge (d);
            \end{scope}
        
            \draw [decorate,decoration = {brace,mirror}] (-1,-1.6) --  (5.3,-1.6);
            \node[] (nodesLabel) at (2,-1.9)  {$W_1$};
            \draw [decorate,decoration = {brace,mirror}] (1.1,-2.1) --  (7.4,-2.1);
            \node[] (nodesLabel) at (4,-2.4)  {$W_2$};
            \draw [decorate,decoration = {brace,mirror}] (3.2,-2.6) --  (9.5,-2.6);
            \node[] (nodesLabel) at (6.5,-2.9)  {$W_3$};
            \draw [decorate,decoration = {brace,mirror}] (5.3,-3.1) --  (11.6,-3.1);
            \node[] (nodesLabel) at (8.5,-3.4)  {$W_4$};
        
        \end{tikzpicture}
        }
        \caption{A worst case instance of 3-TVC for \algstaradv{} where the snapshots repeat every three time steps and time steps 1 and 3 contain all edges appearing in the following ones $\{2,3\}$ and $\{5,6\}$, respectively, while the edge in these time steps are distinct. The green highlighted temporal vertices show the optimal solution and the red ones are in the solution computed by \algstaradv{}.}
        \label{pic:own_adv_worst}
        \vspace*{-0.3cm}
    \end{figure}
    In Figure~\ref{pic:own_adv_worst}, the green vertex appearances show the exact solution, while the red ones are the approximate solution computed by \algstaradv{}. The exact solution contains the star center appearance of every~$G_i$, where~$i\%\Delta =0$. Hence, its size is~$\left\lceil \frac{T}{\Delta}\right\rceil$. The solution of \algstaradv{} on the other hand would exclude the star center appearances of these~$G_i$'s if~$\Delta\ge2$, since they appear first and all their edges can be covered through the other star center appearances in each window, but would include the star center of all other time steps, since the subgraphs in the time steps~$\{j | i<j<i+\Delta \text{ } \ \forall\ i\%\Delta =0 \}$ are distinct in every window. Hence, on instances with the considered topology our algorithm computes a solution of size~$T-\left\lceil \frac{T}{\Delta}\right\rceil$ when~$\Delta\ge2$, what is stated in Theorem~\ref{thm:star_advanced}. By decomposing the lifetime over~$\Delta$, we get~$T = a \cdot \Delta +b$, where~$a, b \in \mathbb{N}^+_0$ and~$b<\Delta$. To get to the ratio, we distinguish two cases,~$b=0$ and~$b>0$. In the first case, we consider~$T = a\cdot\Delta$ ($b=0$). The size of the exact solution is $\left\lceil \frac{T}{\Delta}\right\rceil = \left\lceil \frac{a\cdot\Delta}{\Delta}\right\rceil = a$.
    Therefore the approximation ratio is
    \begin{equation}
        \nonumber
        \frac{T-\left\lceil \frac{T}{\Delta}\right\rceil}{\left\lceil \frac{T}{\Delta}\right\rceil} = \frac{a\cdot\Delta-a}{a}= \Delta-1 \text{ .}
    \end{equation}
    For~$b>0$ we have~$T=a\cdot\Delta+b$ and the size of the exact solution is $\left\lceil \frac{T}{\Delta}\right\rceil = \left\lceil \frac{a\cdot\Delta+b}{\Delta}\right\rceil = a+1$.
    In this second case the ratio can be calculated as
    \begin{equation}
        \nonumber
        \begin{aligned}
            \frac{T-\left\lceil \frac{T}{\Delta}\right\rceil}{\left\lceil \frac{T}{\Delta}\right\rceil} & = \frac{a\cdot\Delta+1-(a+1)}{a+1}                     
            & = \left(1-\frac{1}{a+1}\right)\Delta + \frac{b}{a+1} -1
            & = \Delta-1 -\frac{\Delta-b}{a+1} \text{ .}
        \end{aligned}
    \end{equation}
    Since ~$b<\Delta$ holds, the last subtrahend is always positive. This results in the ratio is always smaller in the second case and the maximal ratio being reached in the first case. Therefore, the approximation ratio of \algstaradv{} is~$\Delta -1$.
\end{proof}

Similar to \algstartriv{}, \algstaradv{} computes the exact solution for~$1$-TVC, since the exact solution contains every appearing star center. Further, the algorithm is also exact for~$2$-TVC, since Theorem~\ref{thm:star_advanced} proves a ratio of~$1$ for~$\Delta=2$.

\section{Experimental Evaluation}
\label{sec:exp}

\newcommand{\algdapprox}{\texttt{D-APPROX}}
\newcommand{\algdsapprox}{\texttt{D-APPROX-S}}
\newcommand{\algdoneapprox}{\texttt{D-1-APPROX}}

We compare the proposed \algstartriv{} and \algstaradv{} algorithms to the known $d$-~\cite{akrida_temporal_2020} and \hbox{$d-1$-~\cite{hamm_complexity_2022}} approximation algorithms, to which we refer as \algdapprox{} and \algdoneapprox{}, respectively. On small instances we also compute an exact solution computed by the algorithm in \cite{hamm_complexity_2022}. All algorithms were implemented in C++20\footnote{\tvcsolverrepo}.
These algorithms have different approximation ratios on always star temporal graphs. To take this into account, we perform experiments under three conditions:~$\Delta<d$, $\Delta>d$ and on larger instances.
Further, for the known algorithms there exist real-life instances that are appropriate to experiment with, which we use to test their general performance.

For the experiments we use SLURM jobs with 8 cores and 100GiB of RAM on an Ubuntu 20.04.5 LTS machine with linux kernel version 5.4.0-135, 112-core Intel(R) Xeon(R) Gold 6238R CPU running at 2.20GHz, and 512GiB main memory. We repeat the experiments three times and build the geometric mean.
In some graphic representations we used one algorithm as a baseline and compare the others in relation to it; the respective axes are labeled as 'relative'.
We describe differences as \textit{improvement}  which is calculated as~$\left(\frac{\sigma_B}{\sigma_A}-1\right)*100\%$~\cite{faraj_buffered_2022}, where Algorithm A is compared with Algorithm B and~$\sigma_S$ is some objective.

To the best of our knowledge, there are no real-life temporal graph databases specifically for the class of always star temporal graphs; therefore, we use artificially generated ones. We developed a generator\footnote{\graphgenrepo} allowing us to bound the maximum degree any star snapshot, to enable a comparison between the always star and the always degree at most~$d$ algorithms.
To push the proposed algorithms to their limits we want to include closer instances to the worst cases. Therefore, in all experiments half of the tested instances are (underlying) star temporal graphs, a subclass of the always star temporal graphs.
For reproduciblity, we state the generator configurations in Table~\ref{tab:exp:graphs:stars:config}. In the experiments with larger instances, we decided to prioritize a long lifetime and fewer nodes, to see the results in that case.

\begin{table}[!bt]
    \begin{center}
        \caption{Used configurations for the artificial always star graph generation}
        \label{tab:exp:graphs:stars:config}
        \resizebox{0.9\textwidth}{!}{
            \begin{tabular}{|l|l|l|l|l|}
                \hline
                Experiment                                   & $|V|$              & $T$              & $d$                             & random seed          \\
                \hline
                $\Delta<d$   & \numprint{128} & \numprint{64}    & \{\numprint{10}, \numprint{15}, \numprint{20}, \numprint{25}, \numprint{30} \}   &  \{\numprint{0}, \numprint{3}\}\\
                $\Delta>d$   &         \numprint{128} & \numprint{64}  &  \{\numprint{3}, \numprint{4}, \numprint{5}, \numprint{6}, \numprint{7}, \numprint{8}\}  &\{\numprint{0}, \numprint{3}, \numprint{5}\}    \\
                Larger instances            &  \{\numprint{1024}, \numprint{4096}\} & \{\numprint{2354}, \numprint{4692}, \numprint{7030}, \numprint{9369}, \numprint{11707}, \numprint{14045}\}   & \{\numprint{5}, \numprint{50}, \numprint{100}\}  &\{\numprint{0}, \numprint{3}, \numprint{5}\} \\
                \hline
            \end{tabular}
           }
    \end{center}
    \vspace*{-0.5cm}
\end{table}

There exist, however, real-life data that are appropriate for testing the known \algdapprox{} and \algdoneapprox{} algorithms against each other; we do so in Section~\ref{sec:reallifedata}. 
For real-life data we use graphs from the SNAP library~\cite{leskovec_graph_2007}, see Table \ref{tab:exp:graphs:reallife}.
We preprocess the data to the format used by our framework.
In terms of ease, we consider hourly contacts.
We remove the direction of the edges, self-loops and any possible other provided information.

\begin{table}[!bt]
    \begin{center}
        \caption{Real-life temporal graph dataset from the SNAP library \cite{leskovec_graph_2007}}
        \label{tab:exp:graphs:reallife}
        \resizebox{0.9\textwidth}{!}{
            \begin{tabular}{|l|l|l|l|l|}
                \hline
                Graph                      & $T$              & $|V|$              & $|E|$              & Description                                       \\
                \hline
                email-Eu-core-temporal     & \numprint{19295} & \numprint{1005}    & \numprint{16064}   & E-mails between users at a research institution   \\
                sx-askubuntu               & \numprint{62732} & \numprint{515280}  & \numprint{455691}  & Comments, questions, and answers on Ask Ubuntu    \\
                sx-mathoverflow            & \numprint{56408} & \numprint{88580}   & \numprint{187986}  & Comments, questions, and answers on Math Overflow \\
                sx-superuser               & \numprint{66560} & \numprint{567315}  & \numprint{714570}  & Comments, questions, and answers on Super User    \\
                wiki-talk-temporal         & \numprint{55690} & \numprint{1140149} & \numprint{2787967} & Users editing talk pages on Wikipedia             \\
                CollegeMsg                 & \numprint{4649}  & \numprint{1899}    & \numprint{13838}   & Messages on a Facebook-like platform at UC-Irvine \\
                soc-redditHyperlinks-body  & \numprint{29184} & \numprint{27862}   & \numprint{137808}  & Hyperlinks between subreddits on Reddit           \\
                soc-redditHyperlinks-title & \numprint{29184} & \numprint{43694}   & \numprint{234777}  & Hyperlinks between subreddits on Reddit           \\
                \hline
            \end{tabular}
            }
    \end{center}
\end{table}

\vspace*{-1em}

\subsection{Experiments under the Condition~$\Delta<d$}
When~$\Delta<d$ the approximation ratios for our proposed algorithms are better than the ratios of the state-of-the-art algorithms. This is also reflected in our experiment, where we computed approximate solution to ~$3$- and~$4$-TVC, see Figure~\ref{fig:exp_star:all}.
This experiment shows that our algorithms provide far better, close to optimal, solutions than \algdapprox{} in this scenario. As expected \algdoneapprox{} performs much better than \algdapprox{}, as the former searches for uncovered triangles, often leading to the detection of the star center. Still, our \algstaradv{} outperforms \algdoneapprox{}.
\begin{figure}[!bt]
    \centering
    \subfloat[\footnotesize{Solution size comparison}\label{fig:exp_star:ratio}]{\includegraphics[width=0.49\textwidth]{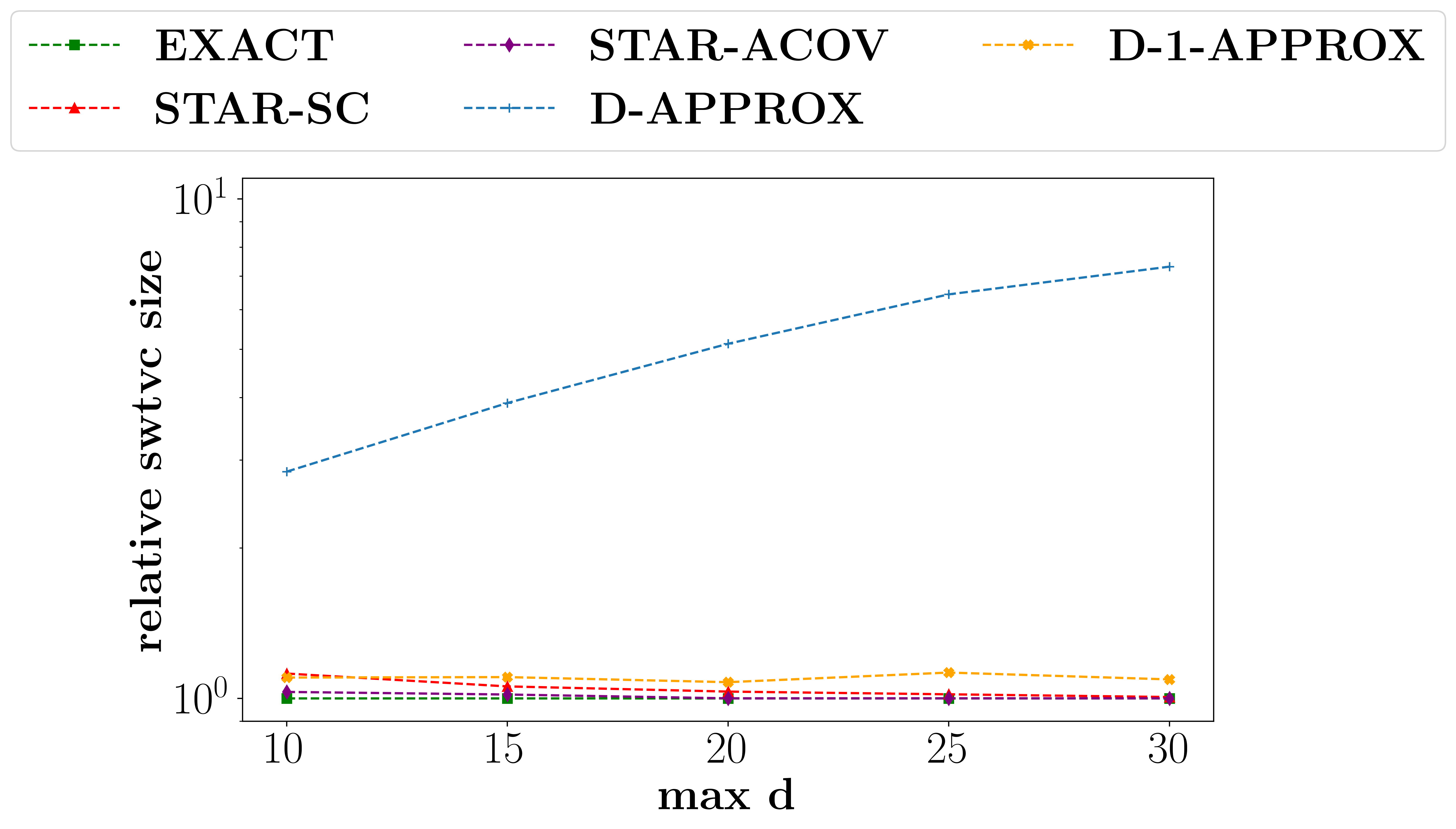}}
    \hfill
    \subfloat[\footnotesize{Running time comparison}\label{fig:exp_star:runtime}]{\includegraphics[width=0.49\textwidth]{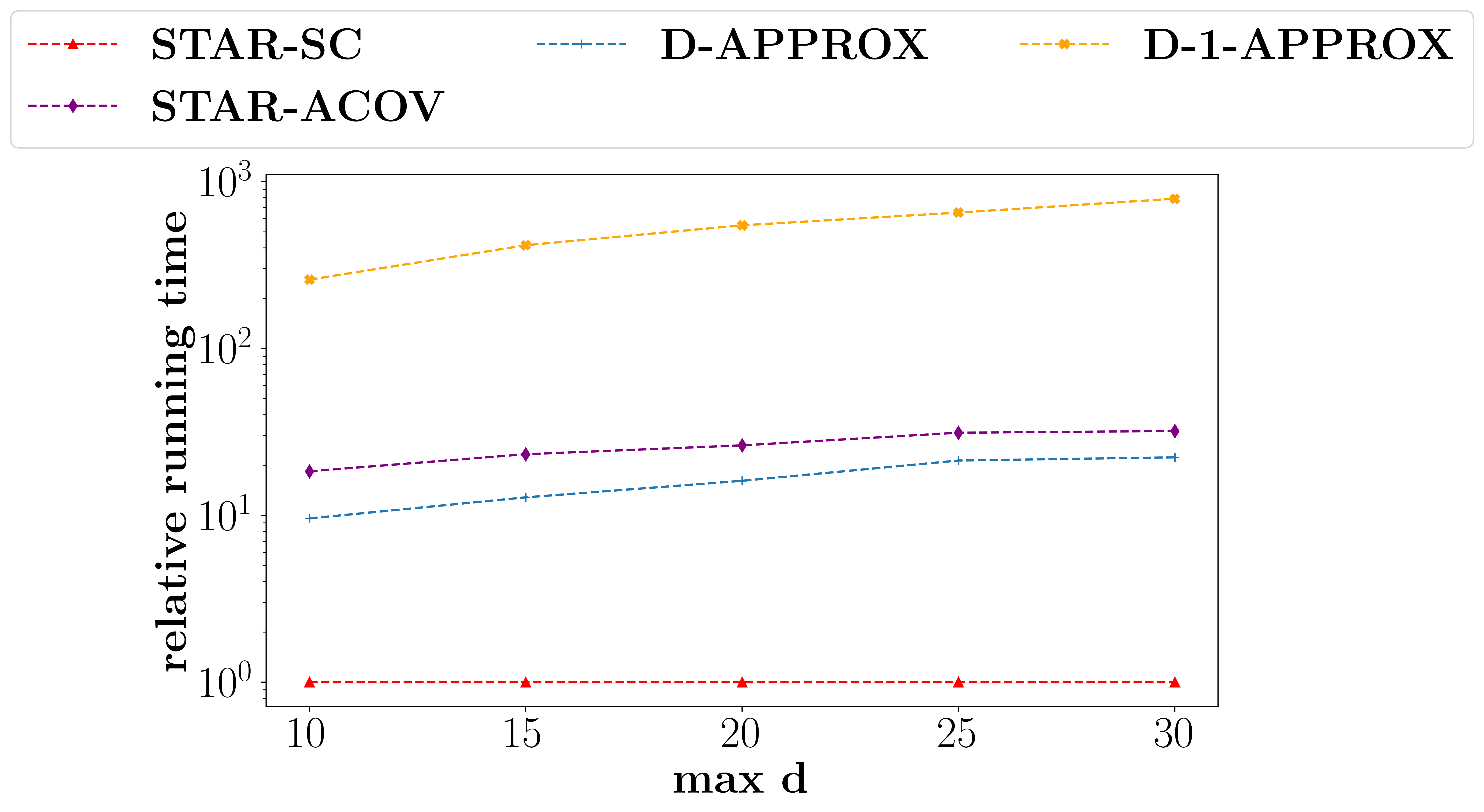}}
    \caption{Averaged results of $3$- and~$4$-TVC comparison of the proposed always star algorithms to always at most $d$ algorithms}
    \label{fig:exp_star:all}
     \vspace*{-0.4cm}
\end{figure}
In terms of running time \algdoneapprox{} takes by far the longest (an average of $28.475$ ms per instance), while the other algorithms are much faster. The average running times are~$0.054$ ms for \algstartriv{},~$0.878$ ms for \algdapprox{} and~$1.404$ ms for \algstaradv{}, which is completely within the expectations since the number of edges in the graphs increases as $d$ increases.
Overall, both \algstartriv{} and \algstaradv{} provide better solutions, and in shorter running time, than \algdoneapprox{}. While \algdapprox{} is faster than \algstaradv{}, its computed solutions are not competitive.
\vspace*{-1em}

\subsection{Experiments under the Condition~$\Delta>d$}
In terms of analysis \algdapprox{} and \algdoneapprox{} yield better worst case ratios than our algorithms when~$\Delta> d$. However, the worst case scenario especially for the \algstaradv{} algorithm is very specific. We argue that in most cases \algstaradv{} still outperforms known approximations, as it is specifically designed for always star temporal graphs and includes at most one vertex in any time step.
Therefore, this experiment tests the algorithms in the case where~$\Delta > d$.
To ensure this in any case, we compute~$20$-TVC.
Figure~\ref{fig:exp_star:smalld:ratio:all} shows that our expectation that \algstaradv{} computes the best results is true.
\algstartriv{} becomes better with the increase of the maximal degree. This can be explained as the increase of the degree of the instances leads to a larger optimal solution, because there are more possible combinations of edges in a window; on the other hand, the solution of \algstartriv{} consists still of every star center and its size stays the same.
\begin{figure}[!bt]
    \centering
    \subfloat[\footnotesize{Solution size comparison}\label{fig:exp_star:smalld:ratio}]{\includegraphics[width=0.49\textwidth]{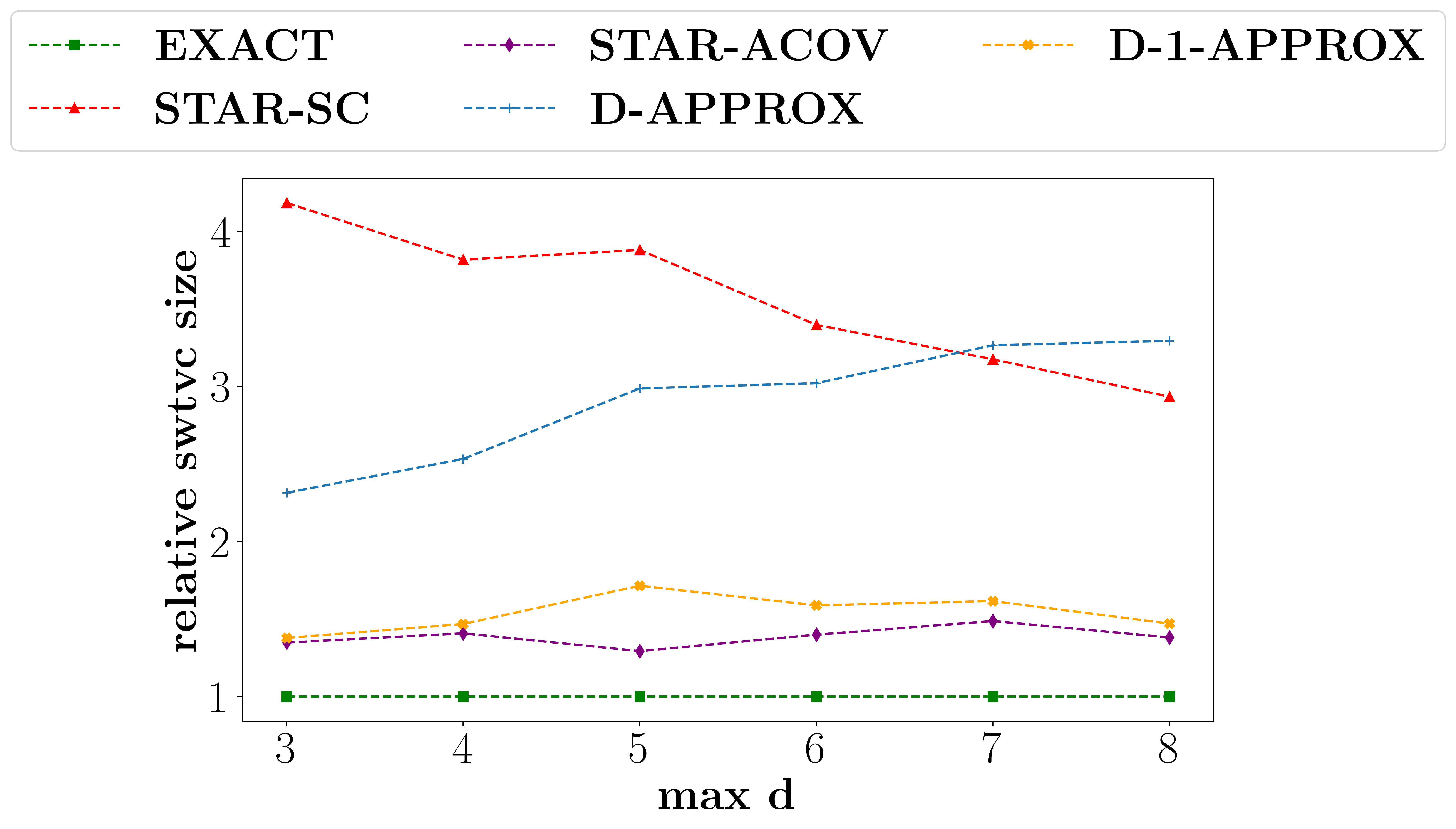}}
    \hfill
    \subfloat[\footnotesize{Running time comparison}\label{fig:exp_star:smalld:runtime}]{\includegraphics[width=0.49\textwidth]{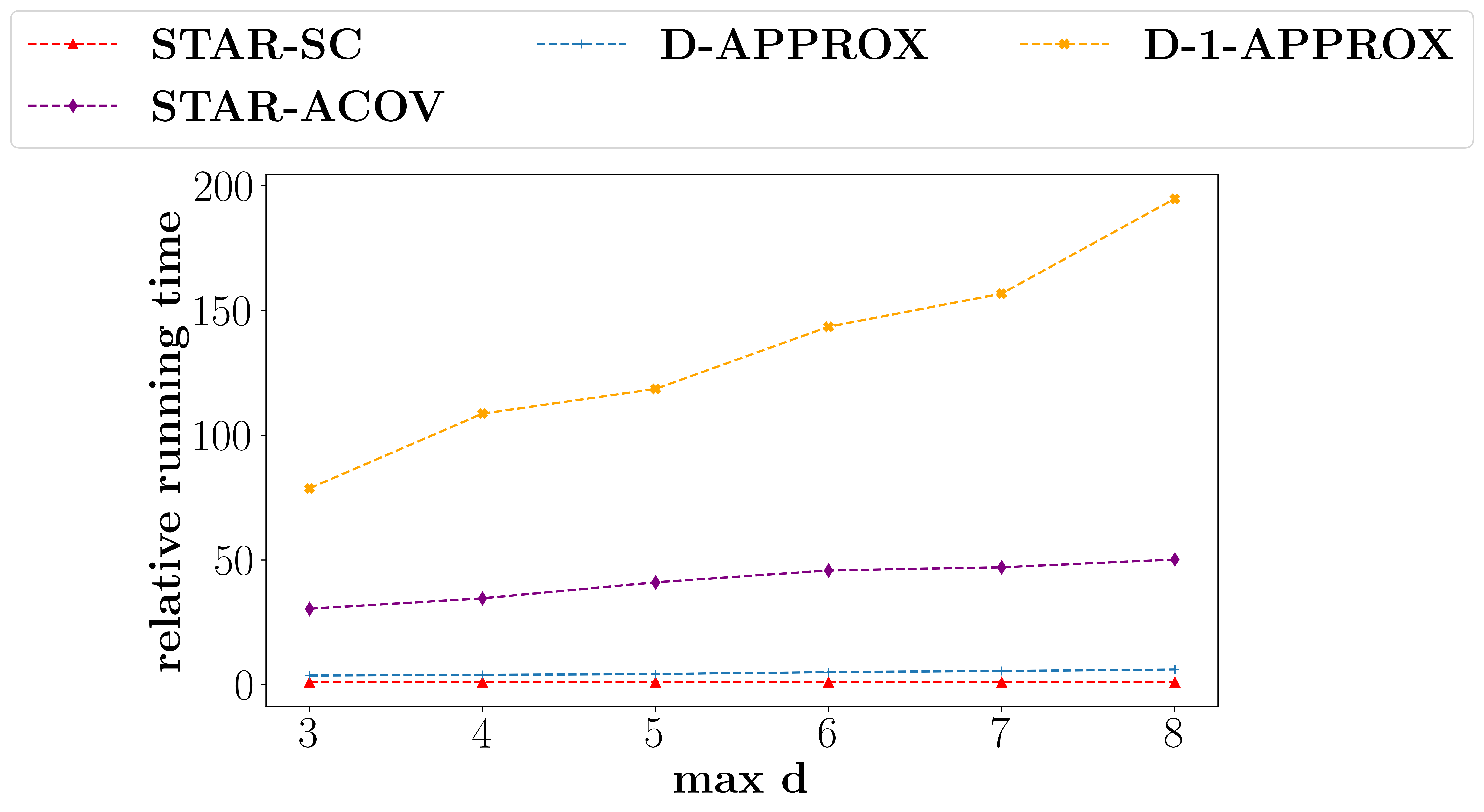}}
    \caption{Averaged results of $20$-TVC comparison of the proposed always star algorithms to always at most $d$ algorithms on small degree graphs}
    \label{fig:exp_star:smalld:ratio:all}
     \vspace*{-0.1cm}
\end{figure}
Similar to before, within the running time expectations, \algdoneapprox{} is the slowest taking averagely~$6.658$ ms per instance. \algstaradv{} follows with~$2.066$ ms, then \algdapprox{} with~$0.237$ ms and finally \hbox{\algstartriv{} with~$0.049$ ms.}
Clearly, \algstaradv{} outperforms \algdoneapprox{} in shorter running time.
However, we know that the approximation ratio of \algstaradv{} is met as shown in the proof of Theorem \ref{thm:star_advanced}. Hence, there exist instances where \algdoneapprox{} will compute better results. It could be interesting for future research to investigate if there is a cut-off point based on the ratio between~$\Delta$ and~$d$ where one algorithm computes smaller solutions. It could also be of interest to collect real life always temporal star instances and check which algorithm performs better on them.
\vspace*{-1em}

\subsection{Experiments on Large Instances}
The experiments in the previous subsection only work with small instances as they provide the exact solution as reference, which is not computable in a reasonable time for larger instances.
Therefore, in the next experiment we compare only the size of the approximations for SW-TVC.
On these inputs we compute~$16$-TVC. The detailed results are displayed in Table~\ref{tab:res_large:detail}.
Figure~\ref{fig:exp_star:large} shows, as expected, that \algdapprox{} performs worst with an average~$16$-TVC size of $\numprint{19582}$, then \algstartriv{} with average of $\numprint{7908}$. \algdoneapprox{} reaches an average size of $\numprint{3441}$, which is again surpassed by \algstaradv{} with a size of $\numprint{2450}$. This is an improvement of 40.46\% of \algstaradv{} compared to \algdoneapprox{}.
\begin{figure}[!bt]
    \centering
    \subfloat[\footnotesize{Solution size comparison}\label{fig:exp_star:large}]{\includegraphics[width=0.49\textwidth]{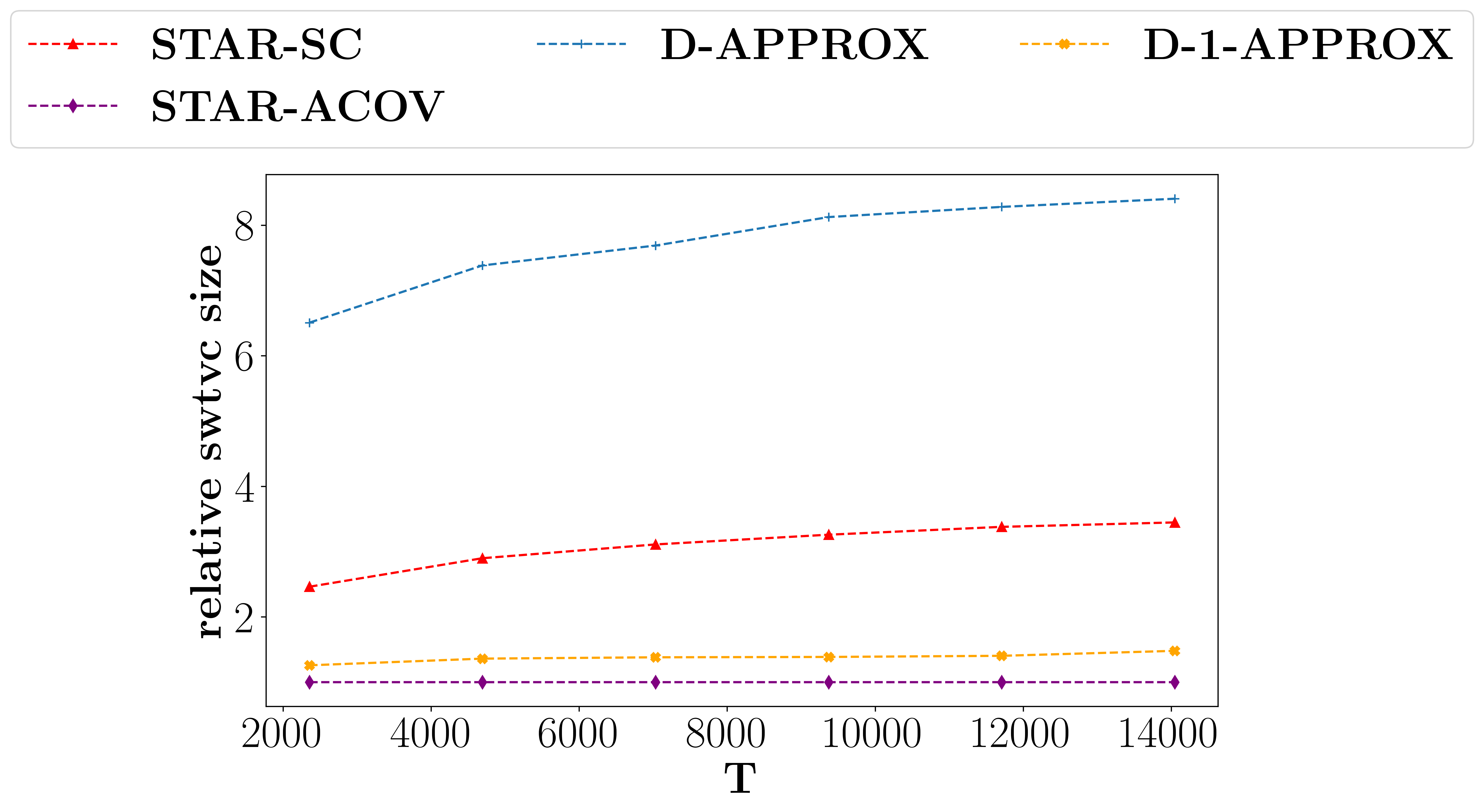}}
    \hfill
    \subfloat[\footnotesize{Running time comparison}\label{fig:exp_star:large:runtime}]{\includegraphics[width=0.49\textwidth]{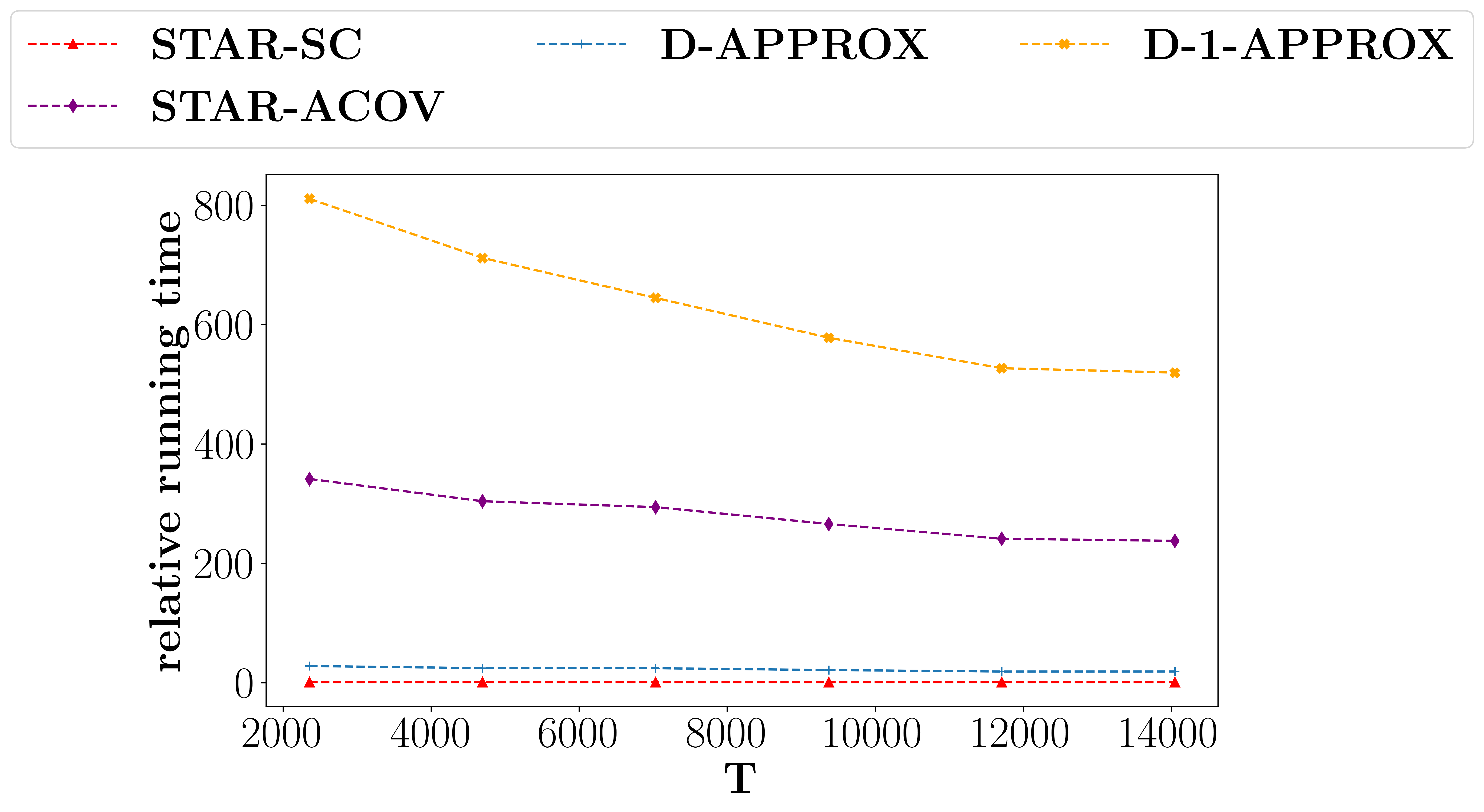}}
    \caption{Averaged results of $16$-TVC comparison of the proposed always star algorithms to always at most $d$ algorithms on larger graphs}
    \label{fig:exp_star:large:all}
     \vspace*{-0.3cm}
\end{figure}
The running times of this experiment is shown in Figure~\ref{fig:exp_star:large:runtime}. \algstartriv{} runs on average $1.821$ ms per instance, followed by the \algdapprox{} running in $245.483$ ms and \algstaradv{} in $3059.587$ ms. The longest running time per instance is needed by \algdoneapprox{} with on average of $6744.801$ ms per instance. Taking the improvement formula from above this leads to a time improvement of \algstaradv{} of 120.44\% against \algdoneapprox{}.

\vspace*{-1em}

\subsection{Experiments of the \algdoneapprox{} and \algdapprox{} on real-life data}\label{sec:reallifedata}

Recall that the previously known algorithms, which provide $d$ \cite{akrida_temporal_2020} and $d-1$ \cite{hamm_complexity_2022} approximations to SW-TVC in $\mathcal{O}(Tm)$ and $\mathcal{O}(T^2m^2)$, respectively, are designed to work for always degree at most $d$ temporal graphs. To the best of our knowledge, these algorithms have not been previously implemented or tested against each other. In this section, we test their performance on the real-life graphs, see Table~\ref{tab:exp:graphs:reallife}.
Figure~\ref{fig:reallife:all} shows the results of a~$64$-TVC computation. The detailed results can be found in Table~\ref{tab:res_reallife:detail}.
As expected \algdoneapprox{} provides better solutions in all cases and achieves an improvement of 11.58\% in solution size.
Rather surprisingly, \algdoneapprox{} is also faster than \algdapprox{} is most cases, achieving a time improvement of 858.35\%. This improvement in running time can be attributed to the fact that all the real-life graphs in the dataset used are sparse in terms of edge appearances. \algdoneapprox{} performs better on sparse graphs as it maintains the current uncovered edge appearances when computing the solution. The algorithm uses those to find uncovered paths of length two where the middle vertex is chosen for the cover.
\algdapprox{} as described in~\cite{akrida_temporal_2020}, on the other hand, checks for an appearance of every edge at any time step and is, therefore, not as well suited for sparse graphs. 

\vspace{-0.5cm}
\subparagraph*{Novel implementation.} We engineered \algdapprox{}~\cite{akrida_temporal_2020} to work faster on sparse graphs and in general, by improving the exact solving of a sparse single edge graph to work on the actual edge appearances instead. For this we exploit the way we store these graphs, i.e., with direct access to $\lambda$ for any edge, and we forward to the next window where the edge actually appears and therefore skip all time steps that do not contain an edge appearance\footnote{\tvcsolverrepo}. We show the results in Figure \ref{fig:reallife} as \algdsapprox. Results clearly show a spectacular improvement, about 260 times, in runtime
with the exact same solution compared to \algdapprox{}.

\begin{figure}[!bt]
    \centering
    \subfloat[\footnotesize{Solution size comparison}\label{fig:reallife}]{\includegraphics[width=0.49\textwidth]{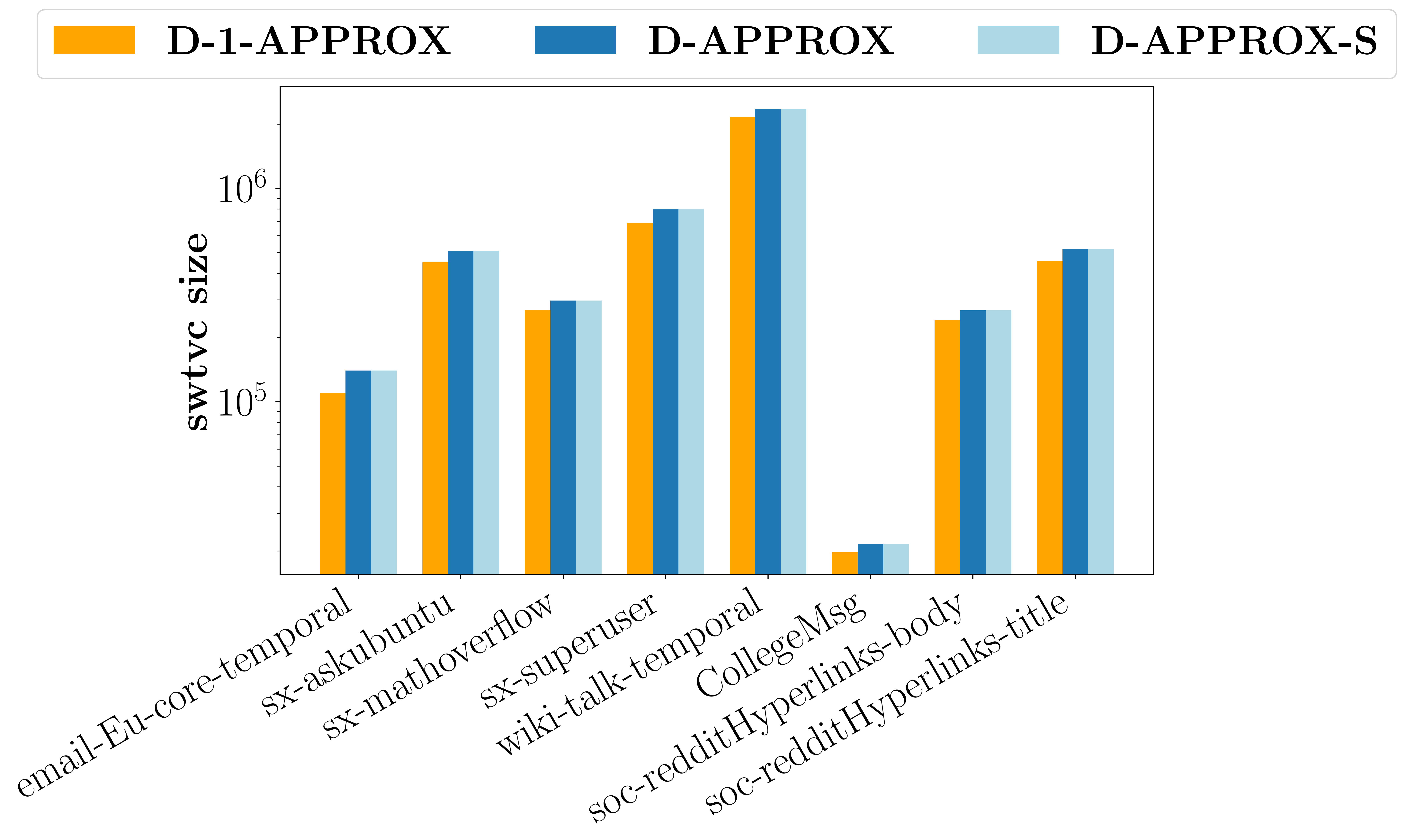}}
    \hfill
    \subfloat[\footnotesize{Running time comparison}\label{fig:reallife:time}]{\includegraphics[width=0.49\textwidth]{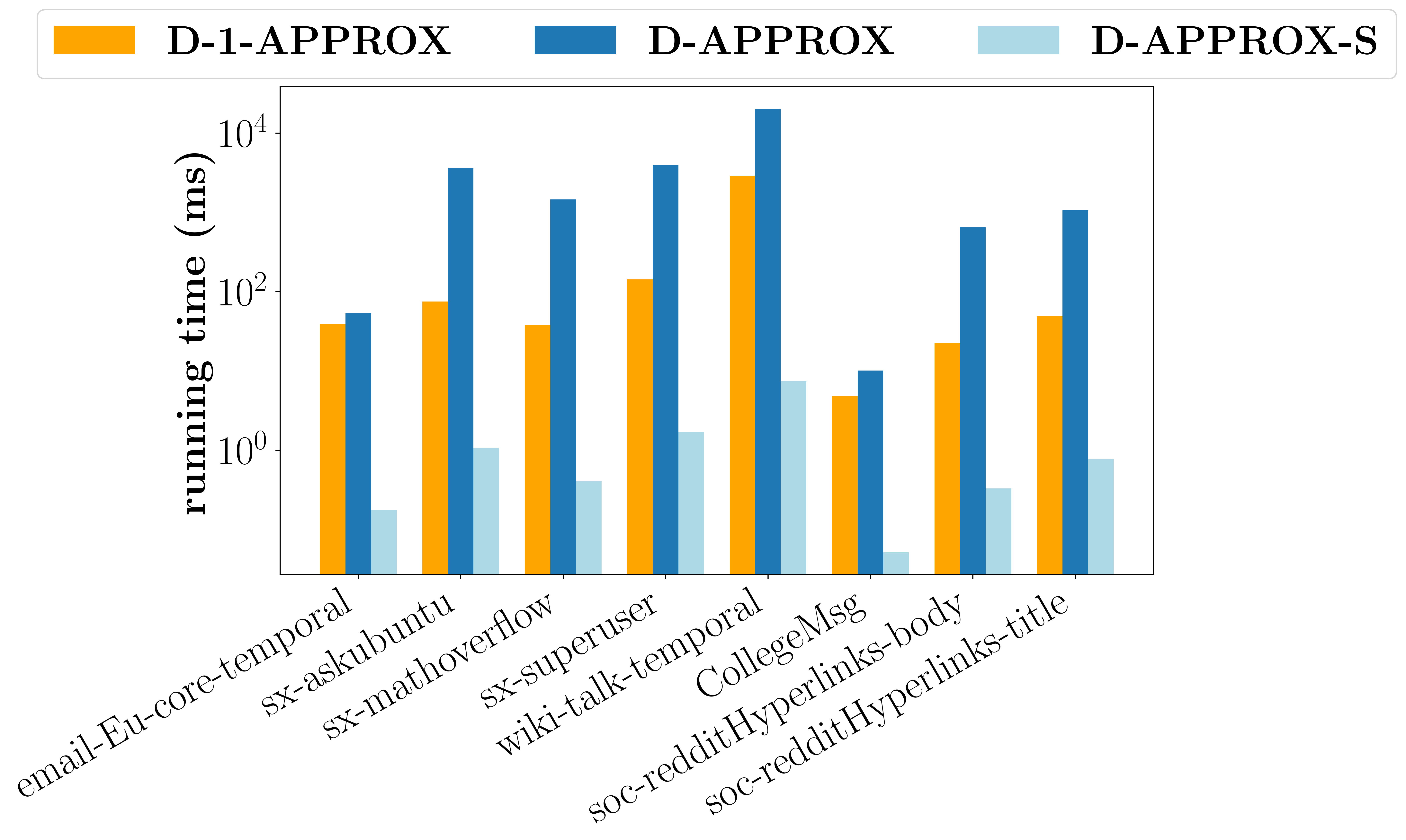}}
    \caption{$64$-TVC comparison of the \algdoneapprox{} and \algdapprox{} on real-life instances}
    \label{fig:reallife:all}
     \vspace*{-0.3cm}
\end{figure}
 
\section{Conclusion}
\label{sec:conclusion}

We introduced two new approximation algorithms for computing~SW-TVC on always star temporal graphs.
The first algorithm, \algstartriv{}, yields a~$2\Delta-1$ approximation ratio in~$\mathcal{O}(T)$ runtime. The second algorithm, \algstaradv{}, yields a~$\Delta-1$ approximation \hbox{ratio in~$\mathcal{O}(Tm\Delta^2)$ runtime.}
The experimental evaluation clearly shows that both proposed algorithms outperform the best previously known algorithm (\algdoneapprox{}) in terms of both solution size and runtime. \algstaradv{} achieves 40.46\% improvement in solution size and 120.44\% improvement in runtime compared to \algdoneapprox{} when computing~SW-TVC with $\Delta=16$ on larger instances.
On non restricted real-life instances, our results show that \algdoneapprox{} outperforms \algdapprox{} in both solution size and runtime. We devise a novel way of implementing the $d$-approximation algorithm, \algdsapprox{}, such that it runs much faster than both \algdoneapprox{} and \algdsapprox{}, while returning the same solution as \algdapprox{}, which is still slightly worse in terms of size than \algdoneapprox{}.

\bibliography{references}

\newpage
\appendix
\section{Detailed Results of the Experiments on Large Always Star Temporal Graphs}
\label{sec:further_res:large}
\setlength{\tabcolsep}{5pt} \renewcommand{\arraystretch}{0.86}
\begin{longtable}{l|rr|rr|rr|rr}
    \caption{This table shows the results for the $16$-TVC on larger always star temporal graphs. We compared the performance of the proposed \algstartriv\ and \algstaradv\ (Algorithm \ref{algo:star_advanced}) ($staradv$) to the known \algdapprox\ and \algdoneapprox\ algorithms. The graph instances are described in the form~$graphclass.maxd.T.n$ where $graphclass\in\{s,u\}$ for $s\widehat{=} star$ and $u\widehat{=} underlying\ star$. The results averaged over the three repetitions and the three instances with the same configuration except for the random seeds, see Table \ref{tab:exp:graphs:stars:config}. $\mid C\mid$ is the size of the computed cover and $t$ the computation time in ms. The smallest cover size and running time are highlighted.
}
    \label{tab:res_large:detail}                                                                                                                                                                                                                                              \\
    \hline
    \multicolumn{1}{c|}{}            & \multicolumn{2}{c|}{\algdapprox} & \multicolumn{2}{c|}{\algdoneapprox} & \multicolumn{2}{c|}{\algstartriv} & \multicolumn{2}{c}{\algstaradv}                                                                                                             \\

    Instances                            & $\mid C\mid$               & $t$                           & $\mid C\mid$                & $t$                         & $\mid C\mid$            & $t$                        & $\mid C\mid$            & $t$                  \\
    \hline
    \hline
s.100.11707.1024  &  \numprint{493114}&\numprint{11215.2} & \numprint{13547}&\numprint{45805.3} & \textbf{\numprint{11706}}&\textbf{\numprint{11.6}} & \textbf{\numprint{11706}}&\numprint{18884.1} \\
\hline
s.50.11707.1024  &  \numprint{263994}&\numprint{6866.2} & \numprint{12828}&\numprint{20848.3} & \textbf{\numprint{11706}}&\textbf{\numprint{11.4}} & \textbf{\numprint{11706}}&\numprint{12212.4} \\
\hline
s.5.11707.1024  &  \numprint{33409}&\numprint{3103.3} & \numprint{12847}&\numprint{1720.6} & \textbf{\numprint{11706}}&\textbf{\numprint{8.8}} & \textbf{\numprint{11706}}&\numprint{5533.0} \\
\hline
s.100.11707.4096  &  \numprint{488701}&\numprint{11175.9} & \numprint{13586}&\numprint{45033.5} & \textbf{\numprint{11706}}&\textbf{\numprint{13.7}} & \textbf{\numprint{11706}}&\numprint{19130.3} \\
\hline
s.50.11707.4096  &  \numprint{261683}&\numprint{6934.5} & \numprint{12874}&\numprint{20778.1} & \textbf{\numprint{11706}}&\textbf{\numprint{11.4}} & \textbf{\numprint{11706}}&\numprint{12080.1} \\
\hline
s.5.11707.4096  &  \numprint{33356}&\numprint{3069.0} & \numprint{12801}&\numprint{1711.8} & \textbf{\numprint{11706}}&\textbf{\numprint{9.2}} & \textbf{\numprint{11706}}&\numprint{5638.1} \\
\hline
s.100.14045.1024  &  \numprint{590472}&\numprint{13690.0} & \numprint{16250}&\numprint{54741.0} & \textbf{\numprint{14045}}&\textbf{\numprint{13.8}} & \textbf{\numprint{14045}}&\numprint{23276.7} \\
\hline
s.50.14045.1024  &  \numprint{316136}&\numprint{8427.9} & \numprint{15402}&\numprint{24591.2} & \textbf{\numprint{14045}}&\textbf{\numprint{13.5}} & \textbf{\numprint{14045}}&\numprint{14632.8} \\
\hline
s.5.14045.1024  &  \numprint{39998}&\numprint{3958.4} & \numprint{15402}&\numprint{2038.0} & \textbf{\numprint{14045}}&\textbf{\numprint{10.6}} & \textbf{\numprint{14045}}&\numprint{6957.6} \\
\hline
s.100.14045.4096  &  \numprint{583799}&\numprint{13686.5} & \numprint{16282}&\numprint{55818.6} & \textbf{\numprint{14045}}&\textbf{\numprint{14.0}} & \textbf{\numprint{14045}}&\numprint{23388.7} \\
\hline
s.50.14045.4096  &  \numprint{312761}&\numprint{8550.1} & \numprint{15431}&\numprint{25304.7} & \textbf{\numprint{14045}}&\textbf{\numprint{13.7}} & \textbf{\numprint{14045}}&\numprint{15095.2} \\
\hline
s.5.14045.4096  &  \numprint{39919}&\numprint{4010.5} & \numprint{15345}&\numprint{2066.3} & \textbf{\numprint{14045}}&\textbf{\numprint{11.2}} & \textbf{\numprint{14045}}&\numprint{6901.8} \\
\hline
s.100.2354.1024  &  \numprint{99774}&\numprint{1657.2} & \numprint{2726}&\numprint{9391.3} & \textbf{\numprint{2354}}&\textbf{\numprint{1.9}} & \textbf{\numprint{2354}}&\numprint{2817.5} \\
\hline
s.50.2354.1024  &  \numprint{53257}&\numprint{973.7} & \numprint{2575}&\numprint{4042.9} & \textbf{\numprint{2354}}&\textbf{\numprint{1.9}} & \textbf{\numprint{2354}}&\numprint{1706.4} \\
\hline
s.5.2354.1024  &  \numprint{6754}&\numprint{358.8} & \numprint{2597}&\numprint{352.0} & \textbf{\numprint{2354}}&\textbf{\numprint{1.8}} & \textbf{\numprint{2354}}&\numprint{640.6} \\
\hline
s.100.2354.4096  &  \numprint{99622}&\numprint{1627.2} & \numprint{2725}&\numprint{9084.7} & \textbf{\numprint{2354}}&\textbf{\numprint{1.9}} & \textbf{\numprint{2354}}&\numprint{2790.7} \\
\hline
s.50.2354.4096  &  \numprint{53243}&\numprint{1020.7} & \numprint{2598}&\numprint{4255.1} & \textbf{\numprint{2354}}&\textbf{\numprint{1.9}} & \textbf{\numprint{2354}}&\numprint{1745.8} \\
\hline
s.5.2354.4096  &  \numprint{6768}&\numprint{355.9} & \numprint{2569}&\numprint{339.3} & \textbf{\numprint{2354}}&\textbf{\numprint{1.8}} & \textbf{\numprint{2354}}&\numprint{624.6} \\
\hline
s.100.4692.1024  &  \numprint{199350}&\numprint{3638.2} & \numprint{5416}&\numprint{18067.7} & \textbf{\numprint{4692}}&\textbf{\numprint{3.9}} & \textbf{\numprint{4692}}&\numprint{6276.6} \\
\hline
s.50.4692.1024  &  \numprint{106503}&\numprint{2359.1} & \numprint{5125}&\numprint{8421.6} & \textbf{\numprint{4692}}&\textbf{\numprint{3.8}} & \textbf{\numprint{4692}}&\numprint{4105.5} \\
\hline
s.5.4692.1024  &  \numprint{13431}&\numprint{944.1} & \numprint{5154}&\numprint{710.2} & \textbf{\numprint{4692}}&\textbf{\numprint{3.6}} & \textbf{\numprint{4692}}&\numprint{1662.9} \\
\hline
s.100.4692.4096  &  \numprint{197276}&\numprint{3802.7} & \numprint{5437}&\numprint{18849.2} & \textbf{\numprint{4692}}&\textbf{\numprint{4.1}} & \textbf{\numprint{4692}}&\numprint{6680.7} \\
\hline
s.50.4692.4096  &  \numprint{105644}&\numprint{2325.9} & \numprint{5162}&\numprint{8374.4} & \textbf{\numprint{4692}}&\textbf{\numprint{3.7}} & \textbf{\numprint{4692}}&\numprint{3940.6} \\
\hline
s.5.4692.4096  &  \numprint{13450}&\numprint{954.0} & \numprint{5129}&\numprint{696.1} & \textbf{\numprint{4692}}&\textbf{\numprint{3.9}} & \textbf{\numprint{4692}}&\numprint{1659.4} \\
\hline
s.100.7030.1024  &  \numprint{298069}&\numprint{5945.0} & \numprint{8109}&\numprint{27045.8} & \textbf{\numprint{7030}}&\textbf{\numprint{6.2}} & \textbf{\numprint{7030}}&\numprint{10145.5} \\
\hline
s.50.7030.1024  &  \numprint{159322}&\numprint{3718.0} & \numprint{7700}&\numprint{12254.5} & \textbf{\numprint{7030}}&\textbf{\numprint{5.6}} & \textbf{\numprint{7030}}&\numprint{6304.5} \\
\hline
s.5.7030.1024  &  \numprint{20109}&\numprint{1586.1} & \numprint{7719}&\numprint{1024.7} & \textbf{\numprint{7030}}&\textbf{\numprint{5.5}} & \textbf{\numprint{7030}}&\numprint{2850.0} \\
\hline
s.100.7030.4096  &  \numprint{293571}&\numprint{5885.5} & \numprint{8162}&\numprint{27007.7} & \textbf{\numprint{7030}}&\textbf{\numprint{5.7}} & \textbf{\numprint{7030}}&\numprint{10522.7} \\
\hline
s.50.7030.4096  &  \numprint{157258}&\numprint{3825.9} & \numprint{7733}&\numprint{12532.5} & \textbf{\numprint{7030}}&\textbf{\numprint{5.5}} & \textbf{\numprint{7030}}&\numprint{6493.1} \\
\hline
s.5.7030.4096  &  \numprint{20032}&\numprint{1639.5} & \numprint{7674}&\numprint{1046.0} & \textbf{\numprint{7030}}&\textbf{\numprint{5.1}} & \textbf{\numprint{7030}}&\numprint{2855.6} \\
\hline
s.100.9369.1024  &  \numprint{396906}&\numprint{8467.7} & \numprint{10837}&\numprint{36032.4} & \textbf{\numprint{9369}}&\textbf{\numprint{9.1}} & \textbf{\numprint{9369}}&\numprint{14456.0} \\
\hline
s.50.9369.1024  &  \numprint{212241}&\numprint{5179.0} & \numprint{10257}&\numprint{16685.3} & \textbf{\numprint{9369}}&\textbf{\numprint{8.8}} & \textbf{\numprint{9369}}&\numprint{8919.8} \\
\hline
s.5.9369.1024  &  \numprint{26800}&\numprint{2386.1} & \numprint{10283}&\numprint{1369.2} & \textbf{\numprint{9369}}&\textbf{\numprint{7.6}} & \textbf{\numprint{9369}}&\numprint{4142.2} \\
\hline
s.100.9369.4096  &  \numprint{392207}&\numprint{8589.3} & \numprint{10876}&\numprint{36938.3} & \textbf{\numprint{9369}}&\textbf{\numprint{9.1}} & \textbf{\numprint{9369}}&\numprint{14512.0} \\
\hline
s.50.9369.4096  &  \numprint{209992}&\numprint{5542.0} & \numprint{10318}&\numprint{17062.7} & \textbf{\numprint{9369}}&\textbf{\numprint{9.2}} & \textbf{\numprint{9369}}&\numprint{9431.2} \\
\hline
s.5.9369.4096  &  \numprint{26738}&\numprint{2316.6} & \numprint{10242}&\numprint{1342.3} & \textbf{\numprint{9369}}&\textbf{\numprint{7.5}} & \textbf{\numprint{9369}}&\numprint{4197.0} \\
\hline
u.100.11707.1024  &  \numprint{27025}&\numprint{60.9} & \numprint{2773}&\numprint{17897.4} & \numprint{11694}&\textbf{\numprint{38.7}} & \textbf{\numprint{1303}}&\numprint{5968.8} \\
\hline
u.50.11707.1024  &  \numprint{14302}&\textbf{\numprint{30.9}} & \numprint{2106}&\numprint{10325.6} & \numprint{11689}&\numprint{37.1} & \textbf{\numprint{1037}}&\numprint{3225.2} \\
\hline
u.5.11707.1024  &  \numprint{1809}&\textbf{\numprint{5.6}} & \numprint{757}&\numprint{2304.2} & \numprint{10881}&\numprint{30.5} & \textbf{\numprint{719}}&\numprint{396.8} \\
\hline
u.100.11707.4096  &  \numprint{5926}&\numprint{53.8} & \numprint{2972}&\numprint{21669.7} & \numprint{11690}&\textbf{\numprint{37.0}} & \textbf{\numprint{1325}}&\numprint{6214.6} \\
\hline
u.50.11707.4096  &  \numprint{4864}&\textbf{\numprint{27.9}} & \numprint{2362}&\numprint{13296.9} & \numprint{11575}&\numprint{33.8} & \textbf{\numprint{1056}}&\numprint{2940.9} \\
\hline
u.5.11707.4096  &  \numprint{1246}&\textbf{\numprint{5.5}} & \numprint{1020}&\numprint{4096.2} & \numprint{10832}&\numprint{29.1} & \textbf{\numprint{720}}&\numprint{408.3} \\
\hline
u.100.14045.1024  &  \numprint{31462}&\numprint{71.4} & \numprint{3504}&\numprint{24929.3} & \numprint{14029}&\textbf{\numprint{49.1}} & \textbf{\numprint{1456}}&\numprint{7255.3} \\
\hline
u.50.14045.1024  &  \numprint{16048}&\textbf{\numprint{35.6}} & \numprint{2739}&\numprint{14396.9} & \numprint{13971}&\numprint{49.0} & \textbf{\numprint{1192}}&\numprint{3641.5} \\
\hline
u.5.14045.1024  &  \numprint{1609}&\textbf{\numprint{6.5}} & \numprint{882}&\numprint{1851.7} & \numprint{10122}&\numprint{27.8} & \textbf{\numprint{679}}&\numprint{353.3} \\
\hline
u.100.14045.4096  &  \numprint{6876}&\numprint{61.0} & \numprint{3376}&\numprint{23772.1} & \numprint{14025}&\textbf{\numprint{48.3}} & \textbf{\numprint{1475}}&\numprint{7340.3} \\
\hline
u.50.14045.4096  &  \numprint{5518}&\textbf{\numprint{32.8}} & \numprint{2883}&\numprint{15264.7} & \numprint{13958}&\numprint{47.3} & \textbf{\numprint{1197}}&\numprint{3618.4} \\
\hline
u.5.14045.4096  &  \numprint{2009}&\textbf{\numprint{6.7}} & \numprint{1327}&\numprint{7310.5} & \numprint{13219}&\numprint{41.4} & \textbf{\numprint{868}}&\numprint{517.5} \\
\hline
u.100.2354.1024  &  \numprint{5600}&\numprint{12.0} & \numprint{1067}&\numprint{4082.3} & \numprint{2352}&\textbf{\numprint{2.7}} & \textbf{\numprint{653}}&\numprint{1178.5} \\
\hline
u.50.2354.1024  &  \numprint{2962}&\numprint{6.4} & \numprint{711}&\numprint{2088.0} & \numprint{2328}&\textbf{\numprint{2.7}} & \textbf{\numprint{442}}&\numprint{627.9} \\
\hline
u.5.2354.1024  &  \numprint{291}&\textbf{\numprint{1.2}} & \numprint{167}&\numprint{167.4} & \numprint{1886}&\numprint{1.9} & \textbf{\numprint{156}}&\numprint{65.3} \\
\hline
u.100.2354.4096  &  \numprint{1597}&\numprint{10.9} & \numprint{993}&\numprint{3623.2} & \numprint{2352}&\textbf{\numprint{2.8}} & \textbf{\numprint{668}}&\numprint{1183.6} \\
\hline
u.50.2354.4096  &  \numprint{1239}&\numprint{5.9} & \numprint{674}&\numprint{2109.7} & \numprint{2333}&\textbf{\numprint{2.8}} & \textbf{\numprint{456}}&\numprint{600.4} \\
\hline
u.5.2354.4096  &  \numprint{273}&\textbf{\numprint{1.3}} & \numprint{222}&\numprint{479.0} & \numprint{2033}&\numprint{2.2} & \textbf{\numprint{167}}&\numprint{85.1} \\
\hline
u.100.4692.1024  &  \numprint{9879}&\numprint{23.8} & \numprint{1563}&\numprint{6094.6} & \numprint{4691}&\textbf{\numprint{9.6}} & \textbf{\numprint{844}}&\numprint{2257.7} \\
\hline
u.50.4692.1024  &  \numprint{4975}&\numprint{14.0} & \numprint{1108}&\numprint{3503.8} & \numprint{4691}&\textbf{\numprint{8.1}} & \textbf{\numprint{582}}&\numprint{1215.1} \\
\hline
u.5.4692.1024  &  \numprint{532}&\textbf{\numprint{2.2}} & \numprint{309}&\numprint{386.2} & \numprint{4132}&\numprint{6.9} & \textbf{\numprint{303}}&\numprint{129.6} \\
\hline
u.100.4692.4096  &  \numprint{2958}&\numprint{21.1} & \numprint{1566}&\numprint{9346.7} & \numprint{4677}&\textbf{\numprint{8.4}} & \textbf{\numprint{866}}&\numprint{2567.6} \\
\hline
u.50.4692.4096  &  \numprint{2423}&\numprint{11.4} & \numprint{1172}&\numprint{6318.8} & \numprint{4667}&\textbf{\numprint{8.3}} & \textbf{\numprint{601}}&\numprint{1335.3} \\
\hline
u.5.4692.4096  &  \numprint{642}&\textbf{\numprint{2.5}} & \numprint{478}&\numprint{1815.4} & \numprint{3836}&\numprint{6.4} & \textbf{\numprint{277}}&\numprint{169.6} \\
\hline
u.100.7030.1024  &  \numprint{16475}&\numprint{35.9} & \numprint{2157}&\numprint{13745.0} & \numprint{7017}&\textbf{\numprint{15.8}} & \textbf{\numprint{1038}}&\numprint{3723.7} \\
\hline
u.50.7030.1024  &  \numprint{8220}&\numprint{17.9} & \numprint{1628}&\numprint{7208.9} & \numprint{6987}&\textbf{\numprint{15.3}} & \textbf{\numprint{749}}&\numprint{1767.3} \\
\hline
u.5.7030.1024  &  \numprint{653}&\textbf{\numprint{3.3}} & \numprint{369}&\numprint{570.5} & \numprint{4810}&\numprint{8.4} & \textbf{\numprint{341}}&\numprint{163.2} \\
\hline
u.100.7030.4096  &  \numprint{4092}&\numprint{30.6} & \numprint{1972}&\numprint{12061.2} & \numprint{7006}&\textbf{\numprint{14.0}} & \textbf{\numprint{986}}&\numprint{3723.9} \\
\hline
u.50.7030.4096  &  \numprint{3246}&\numprint{16.5} & \numprint{1458}&\numprint{7583.9} & \numprint{6904}&\textbf{\numprint{14.4}} & \textbf{\numprint{749}}&\numprint{1900.9} \\
\hline
u.5.7030.4096  &  \numprint{794}&\textbf{\numprint{3.8}} & \numprint{549}&\numprint{1703.2} & \numprint{5695}&\numprint{9.8} & \textbf{\numprint{400}}&\numprint{181.9} \\
\hline
u.100.9369.1024  &  \numprint{20421}&\numprint{46.8} & \numprint{2382}&\numprint{15480.3} & \numprint{9241}&\textbf{\numprint{22.0}} & \textbf{\numprint{1173}}&\numprint{4729.2} \\
\hline
u.50.9369.1024  &  \numprint{10909}&\numprint{24.7} & \numprint{1785}&\numprint{9519.1} & \numprint{9185}&\textbf{\numprint{22.5}} & \textbf{\numprint{869}}&\numprint{2513.2} \\
\hline
u.5.9369.1024  &  \numprint{1179}&\textbf{\numprint{4.4}} & \numprint{478}&\numprint{1478.7} & \numprint{6771}&\numprint{16.3} & \textbf{\numprint{465}}&\numprint{299.7} \\
\hline
u.100.9369.4096  &  \numprint{5087}&\numprint{41.6} & \numprint{2589}&\numprint{17845.6} & \numprint{9306}&\textbf{\numprint{23.5}} & \textbf{\numprint{1169}}&\numprint{4683.4} \\
\hline
u.50.9369.4096  &  \numprint{3922}&\textbf{\numprint{22.0}} & \numprint{2098}&\numprint{10990.7} & \numprint{9297}&\numprint{22.7} & \textbf{\numprint{905}}&\numprint{2423.0} \\
\hline
u.5.9369.4096  &  \numprint{795}&\textbf{\numprint{4.5}} & \numprint{624}&\numprint{1387.7} & \numprint{7566}&\numprint{18.1} & \textbf{\numprint{519}}&\numprint{279.6} \\
\hline
\end{longtable}
\setlength{\tabcolsep}{6pt} \renewcommand{\arraystretch}{1}
 
\section{Detailed Results of the Experiments on Real-Life Instances}
\label{sec:further_res:reallife}
\setlength{\tabcolsep}{4pt} \renewcommand{\arraystretch}{0.86}
\begin{longtable}{l|rr|rr|rr}
    \caption{This table contains the comparison the known \algdapprox\ and \algdoneapprox\ and the engineered \algdsapprox\ algorithms for the $64$-TVC  on real life instances. The results averaged over three runs. $\mid C\mid$ is the size of the computed cover and $t$ the computation time in ms. The smallest cover size and running time are highlighted.}
    \label{tab:res_reallife:detail}                                                                                                                         \\
    \hline
        \multicolumn{1}{c|}{}        & \multicolumn{2}{c|}{\algdapprox} & \multicolumn{2}{c}{\algdoneapprox} & \multicolumn{2}{c}{\algdsapprox}                                                              \\
    
        Graph                        & $\mid C\mid$              & $t$                          & $\mid C\mid$             & $t$  & $\mid C\mid$              & $t$                           \\
        \hline
        \hline
CollegeMsg  &  \numprint{7212}&\numprint{3411.094} & \textbf{\numprint{6564}}&\numprint{1598.437} & \numprint{7212}&\textbf{\numprint{16.703}} \\
\hline
email-Eu-core-temporal  &  \numprint{46693}&\numprint{18354.806} & \textbf{\numprint{36557}}&\numprint{13183.621} & \numprint{46693}&\textbf{\numprint{59.364}} \\
\hline
soc-redditHyperlinks-body  &  \numprint{89494}&\numprint{218184.663} & \textbf{\numprint{80832}}&\numprint{7007.208} & \numprint{89494}&\textbf{\numprint{110.038}} \\
\hline
soc-redditHyperlinks-title  &  \numprint{173843}&\numprint{346342.450} & \textbf{\numprint{153011}}&\numprint{16296.375} & \numprint{173843}&\textbf{\numprint{263.819}} \\
\hline
sx-askubuntu  &  \numprint{169439}&\numprint{1203780.604} & \textbf{\numprint{150097}}&\numprint{26270.006} & \numprint{169439}&\textbf{\numprint{353.491}} \\
\hline
sx-mathoverflow  &  \numprint{99488}&\numprint{471242.940} & \textbf{\numprint{89729}}&\numprint{12575.524} & \numprint{99488}&\textbf{\numprint{139.119}} \\
\hline
sx-superuser  &  \numprint{265889}&\numprint{1301267.243} & \textbf{\numprint{229535}}&\numprint{47844.284} & \numprint{265889}&\textbf{\numprint{557.358}} \\
\hline
wiki-talk-temporal  &  \numprint{785098}&\numprint{6820182.412} & \textbf{\numprint{720962}}&\numprint{956878.618} & \numprint{785098}&\textbf{\numprint{2653.711}} \\
\hline
\end{longtable}
\normalsize

\end{document}